\documentclass[aps,prx,showpacs,twocolumn,reprint]{revtex4-2}

\usepackage{qcircuit}
\usepackage[dvips]{graphicx}
\usepackage{amsmath,amssymb,amsthm,mathrsfs,amsfonts,dsfont}
\usepackage{subfigure, epsfig}
\usepackage{braket}
\usepackage{hyperref}
\usepackage{bm}
\usepackage{enumerate}
\usepackage{color}
\usepackage{graphicx}
\usepackage{pgf}

\usepackage{xcolor}

\usepackage{amsthm}

\newtheorem{result}{Result}
\newtheorem{definition}{Definition}

\newtheorem{theorem}{Theorem}
\newtheorem{lemma}{Lemma}
\newtheorem{example}{Example}

\newcommand{\tr}{\mathrm{Tr}}

\newcommand{\var}{\mathrm{Var}}

\newcommand{\perm}{D_n}

\newcommand{\prob}{\mathrm{prob}_0}

\newcommand{\qn}{Q_n}

\begin{document}


\title{Exponential Error Suppression for Near-Term Quantum Devices}

\author{B\'alint Koczor}
\email{balint.koczor@materials.ox.ac.uk}
\affiliation{Department of Materials, University of Oxford, Parks Road, Oxford OX1 3PH, United Kingdom}


\begin{abstract}
	As quantum computers mature, quantum error correcting codes (QECs) will be adopted in order to suppress errors to any desired level $E$ at a cost in qubit-count $n$ that is merely poly-logarithmic in $1/E$. However in the NISQ era, the complexity and scale required to adopt even the smallest QEC is prohibitive. Instead, error mitigation techniques have been employed; typically these do not require an increase in qubit-count but cannot provide exponential error suppression. Here we show that, for the crucial case of estimating expectation values of observables (key to almost all NISQ algorithms) one can indeed achieve an effective exponential suppression. We introduce the Error Suppression by Derangement (ESD) approach: by increasing the qubit count by a factor of $n\geq 2$, the error is suppressed exponentially as $Q^n$ where $Q<1$ is a suppression factor that depends on the entropy of the errors. The ESD approach takes $n$ independently-prepared circuit outputs and applies a controlled derangement operator to create a state whose symmetries prevent erroneous states from contributing to expected values. The approach is therefore `NISQ-friendly' as it is modular in the main computation and requires only a shallow circuit that bridges the $n$ copies immediately prior to measurement. Imperfections in our derangement circuit do degrade performance and therefore we propose an approach to mitigate this effect to arbitrary precision due to the remarkable properties of derangements. a) they decompose into a linear number of elementary gates -- limiting the impact of noise b) they are highly resilient to noise and the effect of imperfections on them is (almost) trivial. In numerical simulations validating our approach we confirm error suppression below $10^{-6}$ for circuits consisting of several hundred noisy gates (two-qubit gate error $0.5\%$) using no more than $n=4$ circuit copies.
\end{abstract}

\maketitle

\section{Introduction}

The control of errors, also called noise, is fundamental to the successful
exploitation of quantum computers. The powerful and general theory of quantum
fault tolerance, exploiting quantum error correcting codes (QECs), provides a
theoretical blueprint for controlling errors in the era when quantum
devices are  large-scale \cite{nielsenChuang, lidar2013quantum, PhysRevA.57.127,PhysRevA.52.R2493, PhysRevA.54.1098,knill1998resilient,10.1145/258533.258579}.
Encoding qubits into collective states permits the suppression of the error rate
on {\it logical} gates to an arbitrary small level at the cost of increasing the number
of physical qubits.
Below a threshold the error suppression is exponential in the hardware scaling. However, this powerful solution is prohibitive in the current era of noisy, intermediate scale quantum (NISQ) devices for the following reasons \cite{preskill2018quantum}. (a) the qubit-count scale factor is at least $5$
for the simplest codes that protect against comprehensive noise types \cite{PhysRevLett.77.198, PhysRevA.54.3824}.
(b) the extra circuit complexity that is needed in order to monitor the stabilisers, or equivalent measures of code integrity, is very considerable and will boost the effective error rate.
(c) in order to achieve a universal set of quantum operations on code-protected logical qubits, highly-non-trivial additional measures such as magic state purification must be undertaken, greatly increasing the hardware scale. 

Here we present an approach to controlling errors that achieves the key benefit of true QEC in the specific (but pivotal) case of estimating expected values of operators, and does so without the three key drawbacks of QEC mentioned above.
The present idea requires an increased qubit-count (by some integer factor that is at least two),
and therefore it is more hardware-expensive than many NISQ error mitigation schemes~\cite{cerezo2020variationalreview, endo2020hybrid,bharti2021noisy,Li2017,PhysRevX.8.031027, kandala2018extending,PhysRevLett.119.180509, strikis2020learning, czarnik2020error, PhysRevLett.122.180501, rattew2020quantum}, but in return it provides exponential
error suppression -- which other NISQ solutions cannot.
Therefore the approach might be seen as sitting between the established NISQ-era techniques and the
full QEC domain, albeit nearer to the NISQ approaches.
Moreover the present approach is compatible with other NISQ mitigation techniques such as extrapolation, quasi-probability or symmetry
verification~\cite{endo2020hybrid,Li2017,PhysRevX.8.031027, kandala2018extending,PhysRevLett.119.180509, strikis2020learning, czarnik2020error, PhysRevLett.122.180501};
In fact, extrapolation is used in the present analysis to negate the impact of errors in the derangement process.

\subsection{Estimating Expectation Values}

Estimating expectation values on a quantum device is of central
importance and most near-term applications do need to estimate
such expectation values.
Many variants of the so-called variational quantum eigensolver have
been proposed for solving classically intractable problems, such as
simulating quantum systems described by
Hamiltonians $\mathcal{H}$~\cite{cerezo2020variationalreview, endo2020hybrid,bharti2021noisy,farhi2014quantum,peruzzo2014variational,wang2015quantum,PRXH2,PhysRevA.95.020501,mcclean2016theory,
	PhysRevLett.118.100503,Li2017,PhysRevX.8.011021,Santagatieaap9646,kandala2017hardware,kandala2018extending,
	PhysRevX.8.031022,romero2018strategies,higgott2018variational,mcclean2017hybrid,colless2017robust,kokail2019self,sharma2020noise, koczor2019variational, koczor2019quantum, koczor2020quantumAnalytic}.
Expectation values of Hamiltonian operators are typically decomposed
as $\langle\psi_{id} | \mathcal{H} | \psi_{id} \rangle = \sum_k c_k \langle\psi_{id} | P_k | \psi_{id} \rangle$,
where $P_k$ are tensor products of Pauli operators, and we will collectively denote them as
$\sigma \equiv P_k$ in the following.
Various approaches have been proposed for estimating such
expectation values $\langle\psi_{id} | \sigma | \psi_{id} \rangle$ 
using quantum computers \cite{endo2020hybrid,Li2017,XYuan2019,Crawford2019,hadfield2020measurements}.
However, without comprehensive error correction, errors during the state preparation
will contribute a bias as  $\langle\psi_k | \sigma | \psi_k \rangle$ into the result,
where $| \psi_k \rangle$ are erroneous states as shown
in the next section. There exist numerous error mitigation techniques that
potentially reduce the effect of such contributions without increasing the number $N$
of qubits, but at the cost of a significantly increased number of measurements and increased numbers
of circuit variants \cite{endo2020hybrid,Li2017,PhysRevX.8.031027, kandala2018extending,PhysRevLett.119.180509, strikis2020learning, czarnik2020error, PhysRevLett.122.180501}. Note that error mitigation techniques are also limited to
correcting errors in measurements of observables as opposed to QECs.

Here we take a different route and introduce the Error Suppression by Derangement (ESD) approach:
we introduce a high degree of symmetry by preparing $n$ copies of the quantum state
$|\psi \rangle$ and use derangement operators (generalised SWAP operations) to protect collective permutation symmetry.
Most noise events that occur during the imperfect preparation of $|\psi\rangle$ break this permutation symmetry and
they are effectively `filtered out' by the ESD.
We outline a possible construction for such a measurement process in Fig.~\ref{schem}
and thoroughly analyse its properties while supporting our claims with rigorous mathematical proofs.

A crucial element of typical NISQ applications is the accurate estimation of
expectation values of observables. The present approach allows one to exponentially
suppress errors in such estimations and thus enables to push the limits of a vast
number of promising NISQ techniques. Let us name a few potential applications: variants of the variational quantum
eigensolver for, e.g., finding ground states of molecular Hamiltonians in quantum chemistry or
spin model Hamiltonians in materials science; quantum approximate optimisations of graph problems;
quantum machine learning
and beyond. Please refer to the review articles \cite{cerezo2020variationalreview, endo2020hybrid,bharti2021noisy}
and references therein for more examples. Since the present approach is completely general and can be
applied to the estimation of any observable (as discussed above), we will present our results and proofs in complete
generality without explicitly specifying or restricting the observable $\sigma$.

Our construction is certainly
very well suited for NISQ hardware for the following reasons. First, the main computation is modular
as the $n$ copies of the computational state are prepared completely independently.
Second, the derangement circuit that `bridges' the $n$ copies immediately prior to measurements
is sufficiently shallow (as it can be decomposed into a linear
number of primitive gates) and therefore picks up significantly less noise than the 
state-preparation stage. Third, the derangement measurement is highly resilient to noise,
since most error events that occur during the derangement process do not contribute to the result.
Let us now introduce basic concepts and explain the main idea in detail.

\begin{figure}[tb]
	\begin{centering}
		\includegraphics[width=0.45\textwidth]{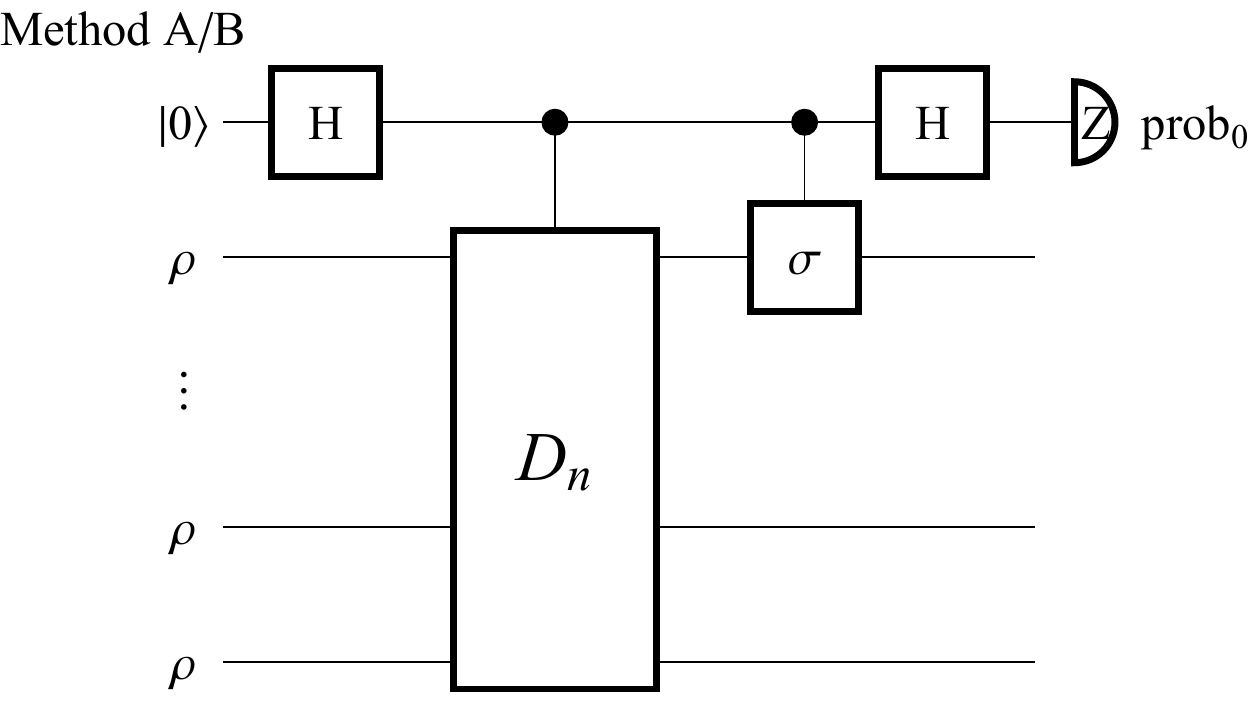}\\[5mm]
		\caption{
			A possible implementation of the ESD approach.
			Our derangement operator $D_n$ is a generalisation of the SWAP operator
			and acts on $n$ (not necessarily identical) copies of the
			quantum state $\rho = \lambda | \psi \rangle \langle \psi  | + (1-\lambda) \rho_{err}$.
			In the above circuit $D_n$ permutes the $n$ registers and allows only
			permutation symmetric combinations, e.g., $| \psi \rangle^{\otimes n}$,
			to contribute to the expectation-value measurement process.
			The probability of measuring the ancilla qubit in the
			$0$ state enables us to approximate the expectation value
			$\langle\psi | \sigma | \psi \rangle$
			and errors are suppressed exponentially in $n$.
			This derangement operator  can be implemented as a shallow circuit using
			a linear number $N(n-1)$ of primitive, controlled, two-qubit SWAP gates. 
			\label{schem}
		}
	\end{centering}
\end{figure}

\section{Preliminaries}

\subsection{Noisy Quantum States and Entropies}

Near-term quantum devices aim to prepare computational
quantum states $| \psi_{id} \rangle$ for, e.g., simulating other quantum systems or beyond.
These quantum devices are, however, imperfect and can only prepare noisy, mixed quantum states
which can be expressed generally via the spectral decomposition of a density matrix	
\begin{equation}\label{rhodef}
\rho = \lambda | \psi \rangle \langle \psi  | + (1-\lambda) \sum_{k=2}^{2^N} p_k | \psi_k \rangle \langle \psi_k  |.
\end{equation}
Here $\lambda \leq 1$ and $\sum_k p_k = 1$ is a probability distribution. 
It is important to recognise that the dominant eigenvector $| \psi \rangle$ above is not necessarily
equivalent to the state that one would obtain from an ideal computation; even 
purely incoherent error models result in a small coherent mismatch in the dominant eigenvector.
A comprehensive analysis of this coherent mismatch is presented in ref.~\cite{koczor2021dominant}
and strong theoretical guarantees are provided that it can be exponentially smaller
than the build-up of the erroneous contributions $ | \psi_k \rangle$.
In the following we thus focus on estimating expectation values in the dominant eigenvector,
while advantages of this approach are discussed below the Acknowledgements.

We further stress that in principle $\lambda$ can be arbitrarily small, e.g., $\lambda=10^{-6}$,
as long as it is the dominant component and larger then any other eigenvalue
as $\lambda > (1-\lambda) p_k$ for all $k$.
Although, for extremely low $\lambda$ other factors such as the sampling cost
may of course become prohibitive in practice as we discuss in later text.

Furthermore, $p_k$ are probabilities of `erroneous' contributions $| \psi_k \rangle$,
and we will refer to these (orthonormal) states as `erroneous' eigenvectors in the following
and we denote their probability vector as $\underline{p}$.
To keep our discussion completely general we do not restrict the probability distribution $\underline{p}$
at all, but we remark that R{\'e}nyi entropies \cite{renyi1961measures} as
$$H_n(\underline{p}) := \frac{1}{1-n} \, \ln [\sum_{k=2}^{2^N} p^n_k  ]$$
will have a crucial effect on the efficacy of the technique and, indeed, for typical experimental quantum systems one can expect
that $H_n(\underline{p})$ are large.

\subsection{Main Idea}

As discussed above, most applications targeting early
quantum devices aim to estimate expectation values  $\langle\psi_{id} | \sigma | \psi_{id} \rangle$
in a quantum state $| \psi_{id} \rangle$ prepared by an ideal noiseless quantum device. Measuring expectation values in the
dominant eigenvector $\langle\psi | \sigma | \psi \rangle$ from Eq.~\ref{rhodef} would give in practical scenarios a very good
approximation \cite{koczor2021dominant}, however, erroneous eigenvectors during state preparation contribute bias $\langle\psi_k | \sigma | \psi_k \rangle$
to the estimated expectation values. Here we aim to suppress these contributions
via the following novel principle. Let us prepare $n$ copies of the state $\rho$
from Eq.~\eqref{rhodef}.
The most likely event during state preparation is that we obtain the dominant eigenvector of the state:
with a probability $\lambda^n$ the resulting state (immediately after state preparation) is $|\psi, \psi,  \dots \psi \rangle$.
Measuring
the expectation value \emph{on the first} register gives the desired result
$\langle \psi, \dots \psi, \psi | \sigma \psi, \psi,  \dots \psi \rangle = \langle\psi | \sigma | \psi \rangle$.

In complete generality, under arbitrary noise models, the second most likely event is that one of the registers,
for example the first register,
is found in the orthogonal erroneous eigenvector of the density matrix $| \psi_k \rangle$;
A measurement then returns the error term
$\langle \psi_k, \dots \psi, \psi | \sigma \psi_k, \psi,  \dots \psi \rangle  = \langle\psi_k | \sigma | \psi_k \rangle$.
However, if one instead measures the expectation value of the product $\sigma \, \mathrm{SWAP}_{1n}$, where
$\mathrm{SWAP}_{1n}$ swaps the registers $1$ and $n$, we then obtain
\begin{equation*}
 \langle \psi_k, \dots \psi, \psi | \sigma \psi, \psi,  \dots \psi_k \rangle=
 \langle\psi_k | \sigma | \psi \rangle \langle\psi_k  | \psi \rangle =0.
\end{equation*}
Here the SWAP operator changed the ordering of the registers 
as $ | \psi_k, \psi,  \dots \psi \rangle  \rightarrow | \psi, \psi,  \dots \psi_k \rangle  $ and
the result is $0$
due to the orthogonality of the eigenvectors of the density matrix.
We can straightforwardly generalise this idea to the case where
all registers are swapped, allowing only permutation-symmetric
states to contribute to the measurement of expectation values. We will refer to this
permutation operation as `derangement'. Let us emphasise that the above argument
is completely general and holds for any	noise model.
While one can certainly realise the
above measurement principle in various different ways, we propose one such
circuit in Fig.~\ref{schem}. We rigorously
prove properties of this particular construction in Result~\ref{result1},
Result~\ref{result2} and Result~\ref{result3}, but we stress that the
current proposal is not limited to the circuit in Fig.~\ref{schem}
(and even Fig.~\ref{schem} leaves room for various different physical
implementations which we discuss in later text).

\section{Results}

\subsection{Exponential Error Suppression}

Let us now formally state the main result of the present work.
In particular, the  circuit in
Fig.~\ref{schem} can be thought of as a Hadamard-test technique \cite{nielsenChuang}
that measures the expectation value of the product $\sigma D_n$, where the derangement operator
$D_n$ permutes the $n$ input registers; as we will explain in a later section
and discuss that it only requires a linear number of primitive gates to construct.
We prove in Theorem~\ref{theo1}
that only permutation-symmetric combinations can pass through the derangement measurement
in Fig.~\ref{schem},
such as the dominant eigenvector $| \psi \rangle^{\otimes n}$
(which happens with a probability $\lambda^n$) or states in which
the same errors occured to all registers $| \psi_k \rangle^{\otimes n}$
(which happen with probabilities $(1-\lambda)^n p_k^n$).
Our general result in Theorem~\ref{theo1} determines the probability $\prob$ of measuring the
ancilla qubit in Fig.~\ref{schem} in the $0$ state as
\begin{align*}\label{prob0}
2 \mathrm{prob}_0 -1 &= \tr[\rho^n \sigma] 
\\ &= \lambda^n \langle\psi |\sigma | \psi \rangle
+  (1-\lambda)^n\sum_{k=2}^{2^N} p_k^n  \langle\psi_k |\sigma | \psi_k \rangle,
\end{align*}
where the erroneous contributions $\langle\psi_k |\sigma | \psi_k \rangle$ are
exponentially suppressed as we increase $n$.

Dividing by $\lambda^n$ allows one to approximate the expectation value of a unitary observable
$\sigma^2 = \mathrm{Id}$, or otherwise the real part of the expected value
of a unitary operator. 
We work out two explicit results in Example~\ref{example1} and Example~\ref{example2}
that demonstrate how the above scheme allows to exponentially suppress the noise
as we increase the number $n$ of copies of $\rho$ and how its efficacy depends
on properties of the probability distribution $p_k$.
Let us now state approximation errors of Methods A and B.

\begin{figure*}[tb]
	\begin{centering}
		\includegraphics[width=0.95\textwidth]{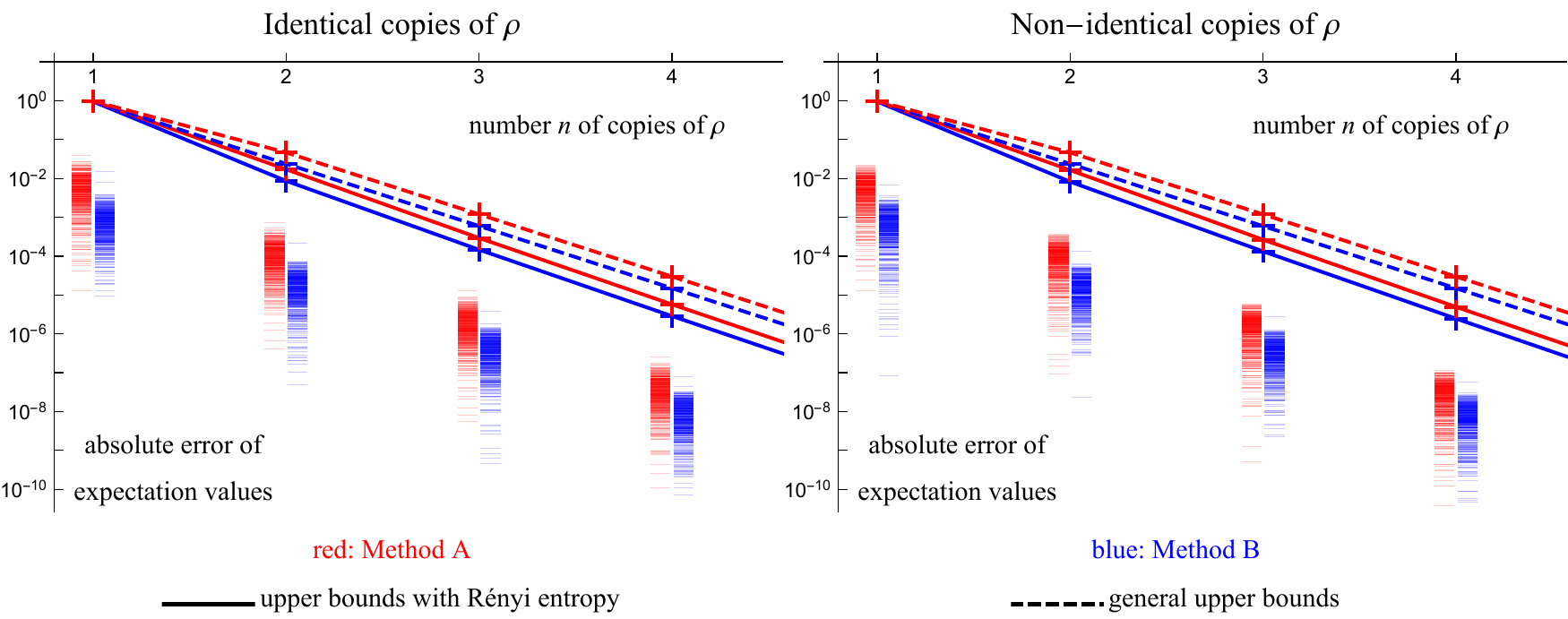}
		\caption{
			Simulation of a $12$-qubit state with $372$ noisy quantum gates.
			Errors in estimating the  expectation value in the dominant eigenvector $\langle\psi|\sigma| \psi \rangle$
			decay exponentially with the number of copies $n$ of the
			quantum state $\rho$. 
			Dashed lines: our general upper bounds on the errors from Result~\ref{result1}
			and Result~\ref{result2} only require the knowledge of the dominant eigenvalue
			$\lambda$ and the largest error probability $p_{max}$ from Eq.~\eqref{rhodef}.
			Solid lines: our upper bounds based on the R{\'e}nyi entropy of the
			quantum states.
			Red and blue bars: approximation errors
			obtained with 500 randomly selected Pauli strings (observables $\sigma$)
			are significantly below the upper bounds.
			(left) all copies of $\rho$ are perfectly identical and (right) all copies of $\rho$ are
			significantly different
			(but they commute as this case can be simulated efficiently) and their trace
			distance is $0.01$ due to their different eigenvalue distributions.
			\label{sampling}
		}
	\end{centering}
\end{figure*}

\begin{result}\label{result1}
		Let us prepare $n$ identical copies of the experimental quantum state
	 $\rho$ from from Eq.~\eqref{rhodef} and apply the derangement measurement
	 from Fig.~\ref{schem}. Both Methods A and B
	 approximate the expectation value $\langle\psi |\sigma | \psi \rangle$
	 by estimating $\prob$ on the ancilla qubit.
	 Method B only estimates $\prob$ and assumes explicit knowledge of the dominant eigenvalue $\lambda$. 
	 In Method A we additionally estimate $\mathrm{prob'}_0$ by repeating the procedure but
	 omitting the controlled-$\sigma$ gate in Fig.~\ref{schem}.
	 We denote their approximation errors as $\mathcal{E}_A$ and $\mathcal{E}_B$, respectively,
	 \begin{align}
	 \text{Method A:}& \quad \quad  \frac{2\mathrm{prob}_0-1}{2\mathrm{prob'}_0-1} = \langle\psi |\sigma | \psi \rangle + \mathcal{E}_A,\\
	 \text{Method B:}& \quad \quad \frac{2\mathrm{prob}_0-1}{\lambda^n} = \langle\psi |\sigma | \psi \rangle + \mathcal{E}_B,
	 \end{align} 
	and these approximation errors generally decay exponentially with the number $n$ of copies via the sequence $Q_n$
	\begin{equation}
	|\mathcal{E}_A| \leq \frac{ 2 \qn}{1+\qn} \quad \text{and} \quad |\mathcal{E}_B| \leq \qn,
	\end{equation}
	which is bounded $\qn \leq  \mathrm{const}\times Q^n$ via the suppression factor $Q < 1$  as established in
	Theorem~\ref{theo2} and in Lemma~\ref{renyilemma}.
\end{result}
Note that these error bounds naturally extend to observables $\mathcal{H}$ of unit
norm that are linear combinations of Pauli strings.
We shown in Lemma~\ref{renyilemma} that the errors also decay exponentially with
the R{\'e}nyi entropy $H_n(\underline{p})$ of the
error probability distribution from Eq.~\eqref{rhodef}  via
$\qn =  (\lambda^{-1}{-}1)^n \exp[ (1{-}n) H_n(\underline{p}) ]$.
Even without knowing or having a good guess of the R{\'e}nyi entropy of the error probabilities,
we can state a general upper bound that only depends on the two largest eigenvalues of the state
as $\qn \leq (\lambda^{-1}{-}1)^n (p_{max})^{n-1}$, where $p_{max}$ is the largest of the
error probabilities $p_k$ in Eq.~\eqref{rhodef}. Note that these quantities, and thus the upper bounds, may
be estimated experimentally \cite{ekert2002direct, PhysRevA.64.052311, marvian2014generalization, 9163139, PhysRevA.89.012117,christandl2007nonzero,christandl2006spectra}.

\subsection{Numerical Simulations}
Let us now numerically verify the above bounds in a practical setting:
We consider a $12$-qubit quantum state that is produced by a noisy, parametrised
quantum circuit typically used in variational quantum algorithms -- our circuit consits
of 10 alternating layers and overall $372$ quantum gates. Refer to Sec.~\ref{numerics}
for more details. Each two-qubit gate undergoes
2-qubit depolarising noise with $0.5\%$ probability and each single-qubit gate undergoes
depolarising noise with $0.05\%$ probability.
The resulting state has a dominant eigenvalue $\lambda \approx 0.51$
and it has a high entropy, full-rank error probability distribution via the R{\'e}nyi entropies
that monotonically decrease with $n$ as
$H_2(\underline{p}) = 4.69$, $H_3(\underline{p}) = 4.38$, $H_4(\underline{p}) = 4.23$,
and $H_\infty(\underline{p}) = 3.63$. Refer to Appendix~\ref{numerics} for more details.

Let us remind the reader that despite the purely incoherent error model, the dominant
eigenvector $| \psi \rangle$ of $\rho$ is slightly different than
what one would obtain from a completely error-free computation 
and in Fig.~\ref{sampling} we compute errors using the dominant eigenvector,
refer to Appendix~\ref{numerics} for more details.

In Fig.~\ref{sampling} (left) we plot our error suppression upper bounds from Result~\ref{result1},
i.e., solid lines represent the error bounds computed from the R{\'e}nyi entropy of the
quantum state's error-probability distribution and dashed lines represent the general
upper bound $\qn \leq (\lambda^{-1}{-}1)^n (p_{max})^{n-1}$ 
where the largest error probability is $p_{max}=0.026$
and the suppression factor is $Q = 0.026$.
Red and blue colours correspond to Method A and Method B, respectively.
We have generated 500 Pauli strings as observables randomly and computed
the errors in estimating their expectation values (there are overall $4^{12} = 1.68 \times 10^7$ Pauli strings,
and we randomly select $500$). These samples
(see horizontal lines in Fig.~\ref{sampling}) are significantly below our upper
bounds and seem to decrease in a similar exponential order
as our bounds (i.e., slope is similar in the logarithmic plot). 

Method B slightly outperforms Method A (slightly smaller errors as blue is slightly below red),
but it requires an exact (or very precise) knowledge of the dominant eigenvalue
$\lambda$. Nevertheless, this eigenvalue could be determined precisely by existing
approaches in special cases, e.g., as in \cite{harper2020fast}.

\subsection{Effect of Non-Identical States}

We now turn to the question of
how the efficacy of our error suppression scheme is affected when the $n$
copies of the state $\rho$ are not identical.
\begin{result} \label{result2}
	We assume that all copies of the quantum state are arbitrarily different via
	$\rho_1 \neq \rho_2 \neq \dots \rho_n$ except that their dominant eigenvector
	is $| \psi \rangle$. Our scheme via Lemma~\ref{differentlemma}
	still provides exponentially decreasing approximation errors when the
	dominant eigenvalue of the worst quality copy is $\lambda_\mathrm{min}>1/2$
	via
	\begin{align}
	\text{Method A:}& \,   \frac{2\prob-1}{2\mathrm{prob'}_0-1} = \langle\psi |\sigma | \psi \rangle + \mathcal{O}([\lambda_\mathrm{min}^{-1}{-}1]^n),\\
	\text{Method B:} &\,  \frac{2\prob-1}{\prod_{\mu=1}^n \lambda_\mu} = \langle\psi |\sigma | \psi \rangle +\mathcal{O}([\lambda_\mathrm{min}^{-1}{-}1]^n).
	\end{align}
	In the special case when all copies of the quantum state commute
	(same eigenvectors but different eigenvalues) one can expect very similar approximation
	errors to Result~\ref{result1} via an effective sequence $Q_n^{eff}$.
\end{result}

We can efficiently simulate the case when all copies of the quantum state commute.
We disturbed every copy of the
density matrix such that their trace distance is $\lVert \rho_k -\rho_l \rVert \approx 10^{-2}$
for all $k\neq l$. Note that the approximation errors in Fig.~\ref{sampling} (right) are very similar
to Fig.~\ref{sampling} (left) and they are approximately upper bounded by the same upper bounds
from Result~\ref{result1} (as expected from Lemma~\ref{differentlemma}).

\subsection{Complexity Analysis}

Let us now analyse resource requirements of
our ESD approach.
In particular, one needs to prepare a suitable number $n$ of
copies of $\rho$ in order to suppress its errors below a threshold level,
that we will refer to as precision and denote as $\mathcal{E}$.
The overall number of qubits required is then $n N+1$, where $N$
is the number of qubits in the computational state $\rho$.
Furthermore, one needs to repeat measurements many times to sufficiently reduce the
effect of so-called shot noise, i.e., we estimate the probability only from a finite
number of repetitions \cite{van2020measurement}. We denote the number of repetitions
as $N_s$.
Let us now summarise our general results
from Lemma~\ref{complexity_lemma}.
\begin{result}\label{result3}
	In order to reach a precision $\mathcal{E}$ in determining 
	the expectation value $\langle\psi |\sigma | \psi \rangle$, one
	requires a logarithmic number
	$n= \mathcal{O}(\ln\mathcal{E}^{-1}/\ln Q^{-1})$ of
	copies of the quantum state $\rho$ (up to rounding).
	Here $Q<1$ is the suppression factor from Result~\ref{result1}
	that depends on R{\'e}nyi entropies.
	The number $N_{s}$ of measurements required to suppress shot
	noise below
	the threshold $\mathcal{E}$ grows polynomially as
	\begin{align*}
	\text{Method A:}& \quad \quad N_{s} = \mathcal{O}[ \mathcal{E}^{-2 (1+2f) } ], \\
	\text{Method B:}& \quad \quad N_{s} = \mathcal{O}[ \mathcal{E}^{-2 (1+f) } ],
	\end{align*}	
	where $f = \ln [\lambda^{-1} /\ln Q^{-1}]$ increases the polynomial
	order compared to the standard shot-noise limit $\mathcal{O}(\mathcal{E}^{-2})$
	and we have derived a general upper bound on $f$ in Lemma~\ref{complexity_lemma}. 
\end{result}

Dividing by the exponentially attenuated factor $\lambda^n$ in both Methods A and B,
in Result~\ref{result1} requires an increasingly large number of measurements
to sufficiently suppress shot noise. Methods A and B are therefore less efficient than
permitted by the standard shot noise limit $N_s = \mathcal{O}(\mathcal{E}^{-2})$.
For example in the extreme, but still valid, case of $\lambda = 10^{-6}$
and $Q=1/2$ we obtain $f = 19.9$ which increases the sampling costs prohibitively in practice.
Nevertheless, the polynomial order of $\mathcal{O}(\mathcal{E}^{-1})$ is only
logarithmically increased via $f$ and its effect might be negligible in
practically relevant scenarios. For example in our simulations in Fig.~\ref{sampling} we
obtain $f =  0.18$ using our expression $Q = (\lambda^{-1}{-}1) p_{max}$ in Lemma~\ref{renyilemma}.
Indeed, we recover the standard shot-noise limit $\mathcal{O}(\mathcal{E}^{-2})$
for very good quality states $\lambda \approx 1$ or for very high entropy probabilities.

In summary, the complexity of our ESD approach only depends on the largest eigenvalue
$\lambda$ of the state and on the suppression factor $Q$ from
Result~\ref{result1} --  which is determined by the R{\'e}nyi entropy of the error probabilities.
As expected, the number $N_s$ of samples grows polynomially with the target precision
$\mathcal{E}^{-1}$ and the system size{(via $n$) grows logarithmically with $\mathcal{E}^{-1}$.
Let us remark that in case of certain applications a global prefactor in observable
expectation values does not matter
-- such as in case of VQE optimisations -- and one can use method B but omitting the division by $\lambda^n$. 
Using Method B significantly reduces the measurement costs and reduces errors from Result~\ref{result1}
when compared to Method A.

\subsection{Derangements of Quantum Registers}

Let us now discuss how to implement derangement circuits using a linearly growing number (in
$n$ and $N$) elementary gate operations. In particular, our
ESD circuit in Fig.~\ref{schem} uses a
generalisation of the SWAP operator that permutes subspaces of quantum registers.
Recall that in general there exist $n!$ permutations of a set of $n$ ordered elements.
Derangements are a subset of the collection of all permutations:
they permute the $n$ elements such that no element remains in place~\cite{roberts2009applied,sagan2013symmetric}.
We define $D_n$ in Definition~\ref{derangement_def} as unitary representations such that they permute
subspaces of $n$ quantum registers.
For example, for $n=2$ our $D_2$ reduces to the usual SWAP operator as
\begin{equation} \label{n2equation}
D_2 |\psi_1, \psi_2\rangle = \mathrm{SWAP}_{12} \, |\psi_1, \psi_2\rangle = |\psi_2, \psi_1\rangle.
\end{equation}
Note that here $\mathrm{SWAP}_{12}$ swaps the two registers, but it
decomposes into $N$ elementary SWAP operations between pairs of qubits within the registers.
For $n=3$ we have two distinct constructions for possible $D_3$ derangement operators as 
\begin{align*}
&\mathrm{SWAP}_{13} \,  \mathrm{SWAP}_{12} \, |\psi_1, \psi_2, \psi_3 \rangle = |\psi_3, \psi_1, \psi_2\rangle,\\
&\mathrm{SWAP}_{23} \, \mathrm{SWAP}_{12} \,  |\psi_1, \psi_2, \psi_3 \rangle = |\psi_2, \psi_3, \psi_1\rangle.
\end{align*}

For $n=4$ one has $6$ possibilities while in general there are $(n-1)!$ possibilities
for constructing distinct derangement operators -- but choosing any one of these
constructions is sufficient for our scheme to work. Indeed, one could construct derangements $D_n$ straightforwardly as
cyclic shifts \cite{ekert2002direct}, but
the large number of possibilities might offer more preferable
constructions that take into account, e.g,  hardware constraints
such as connectivity.
\addtocounter{footnote}{1}
\footnotetext[\value{footnote}]{
	\allowdisplaybreaks
	Please refer to the webpage
	\href{https://qtechtheory.org/derangement_circuits}{[link]}
	and to the repository \cite{git_derangement_circuits}
	for the demonstration material
}
\newcounter{footcombined}
\setcounter{footcombined}{\value{footnote}}
Please refer to~\cite{Note\thefootcombined,git_derangement_circuits}
for illustrations of the corresponding circuits.
Furthermore, we discuss in Appendix~\ref{symmetries} that the large number of symmetries 
in the derangement circuit can be exploited in order to, e.g., reduce
errors that happen during the controlled-SWAP operations.

Regarding gate complexity, derangement operators can be
implemented efficiently in general using $N(n-1)$ elementary
controlled two-qubit SWAP gates, where $N$ is the number of
qubits in the register $|\psi\rangle$ and $n$ is the number of
copies  of $|\psi\rangle$. These minimal SWAP circuits (which optimally
	implement derangement operators) can be constructed
by mapping the corresponding permutations to graph trees \cite{denes1959representation},
refer to Definition~\ref{derangement_def}.

It is important to recognise that while the number of elementary controlled-SWAP gates grows
as $\mathcal{O}(N)$, preparing the quantum state $|\psi\rangle$ generally requires
$\mathcal{O}[a(N) N]$ gates, where $a(N)$ is the
depth of the computation. It is generally expected that for practical problems one needs to go beyond
constant-depth circuits such that the number of gates in the main computation grows faster than $\mathcal{O}(N)$~\cite{bravyi2020quantum,bravyi2018quantum,niu2019optimizing,zhou2020quantum,babbush2018low}. 
Thus the gate count of the derangement circuit can be expected to be of diminishing relative significance
when scaling up computations.
Even if the controlled-SWAP operator is not a hardware-native gate, one
needs at most $6$ native entangling gates to implement the elementary controlled-SWAP operator,
refer to Table~\ref{recompile} in the Appendix.
We demonstrate this below on a practical example assuming a hardware-native gateset
and also briefly discuss connectivity constraints.

\section{Noise Robustness and Limitations}
\subsection{Mitigating Experimental Imperfections}

So far we have assumed that the
derangement operator in Fig.~\ref{schem} is perfect. Indeed, gates involved here are expected
to be noisy in a realistic
scenario which ultimately limits the precision of our approach and increases its complexity.

We show in Example~\ref{error_example} quite generally that the
derangement operator is highly resilient to experimental imperfections
and protects permutation symmetry even under experimental noise.
This is nicely illustrated in our simulated noisy circuit: the unmitigated errors in determining $\prob$ in Fig.~\ref{extrapolation} are quite low and
are below $10^{-2}$ for all $50$ randomly selected states.
The simulated circuit consist of $13$ qubits, i.e, $3$ copies
of a $4$ qubit state, and elementary controlled-SWAP gates undergo 3-qubit depolarisations
with a probability $3 \times 10^{-3}$.
Refer to Appendix~\ref{numerics} for more details

Most importantly, we show  in Example~\ref{error_example} quite generally that most errors
that occur during the derangement measurement
will only trivially affect the final result by (almost) linearly attenuating the
output probability $\prob$ which can in principle be corrected by an
extrapolation.
We use extrapolation techniques~\cite{endo2020hybrid,Li2017,PhysRevX.8.031027, kandala2018extending,PhysRevLett.119.180509,cai2020multi} which typically estimate 
$\prob(\epsilon)$ at different values of $\epsilon$ and extrapolate, e.g., linearly, to
zero noise $\epsilon = 0$. Due to the high degree of noise resilience of the derangement operator, 
the measurement probabilities $\prob(\epsilon)$ are closely approximated
by a linear function in $\epsilon$ and Fig.~\ref{extrapolation}
illustrates that indeed a linear extrapolation surprisingly well approximates the ideal
probability with errors less than $10^{-4}$. 

Here we aim to suppress errors arbitrarily by accounting for the slight non-linearity of the
function $\prob(\epsilon)$.
We prove in Theorem~\ref{pade_theo} quite generally that expectation values
are \emph{exactly} described by degree-$\nu$ polynomials
as $\prob(\epsilon) = \sum_{k=0}^\nu c_k \epsilon^k$
and $\nu$ is the number of noisy gates. It follows that one can in principle determine
the ideal probability by determining
$\prob(\epsilon)$ at $\nu+1$ different values of $\epsilon$
and fitting a degree $\nu$ polynomial.
Fig.~\ref{extrapolation} (blue circles) demonstrates how the
extrapolation error decreases exponentially with the degree of the fitted
polynomial.

\begin{figure}[tb]
	\begin{centering}
		\includegraphics[width=0.48\textwidth]{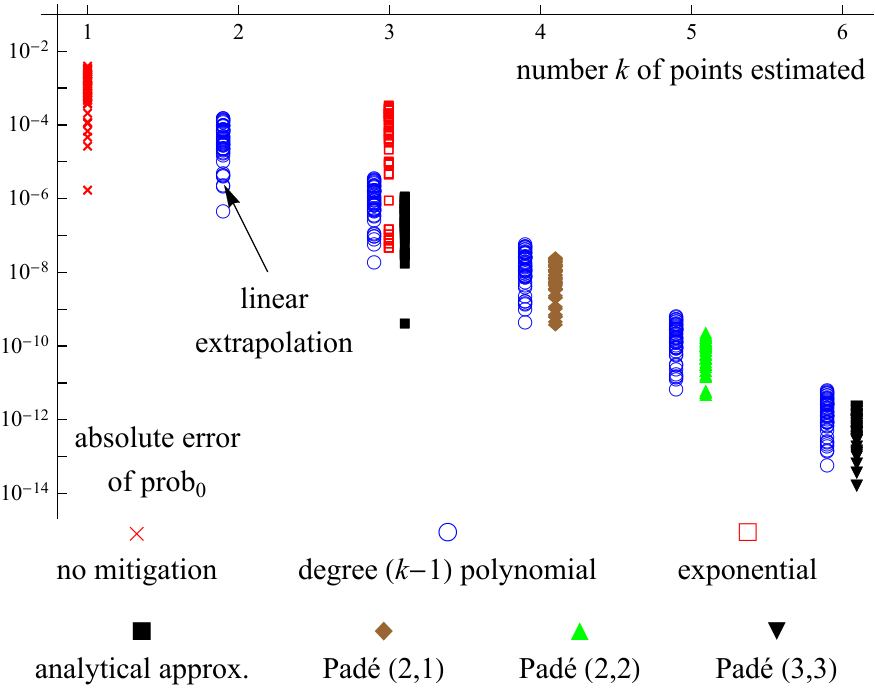}
		\caption{Mitigating errors in the derangement operator.
			Extrapolation errors using various different fitting techniques
			vs. the number of fitting points for $50$ randomly selected
			ansatz states. 	
			Elementary gates in the derangement operator in Fig.~\ref{schem} 
			have an error rate $\epsilon=10^{-3}$ and an experimentalist
			can increase this error in $k=2,3,4\dots$ steps up to $\epsilon=10^{-2}$.
			The probability $\prob$ from Fig.~\ref{schem} at $\epsilon=0$ is estimated 
			by extrapolating to $\epsilon=0$.
			The derangement measurement is highly resilient to imperfections (see text)
			and $\prob(\epsilon)$ is almost linear in $\epsilon$.
			Increasing the degree of the fitting polynomial (blue circles) reduces
			the extrapolation error exponentially.
			\label{extrapolation}
		}
	\end{centering}
\end{figure}

Furthermore, we analytically solve the dependence on $\epsilon$ in the limiting
case of a large number of gates and obtain the approximation
\begin{equation*}
\prob(\epsilon) \approx  \prob  - \tilde{\eta} \epsilon	\, \frac{(1 - \epsilon)^{ \nu} }{ 2  \epsilon-1 }
\approx	\frac{a_1 \epsilon  + a_2 \epsilon^2 + a_3 \epsilon^3
}{ 1 + a_4 \epsilon + a_5 \epsilon^2},
\end{equation*}
where $\tilde{\eta}$ is a constant. The above $(3,3)$ Pad{\'e}
approximation of the analytical dependence can be
determined by fitting the coefficients $a_1, a_2, a_3, a_4, a_5$.
These Pad{\'e} approximations appear to slightly outperform
degree-$k$ polynomial extrapolations in Fig.~\ref{extrapolation}.
Refer to Theorem~\ref{pade_theo} for more details.

In summary, guided by analytical arguments in Example~\ref{error_example}
we propose an efficient and straightforward approach to mitigate
experimental errors that occur during the derangement circuit.
Although in realistic scenarios an experimentalist
may not be able to perfectly amplify all errors, we demonstrate below
that extrapolation techniques can still significantly reduce the impact of noise.
Note, however, that for an increasing number of qubits
the noise in the controlled-SWAP gates accumulates and might
attenuate the output probability $\prob$. Estimating this attenuated
probability at increased error rates---as required for extrapolation---requires
an increased number of measurements.
For example, a factor of $0.1$ attenuation threshold could be approximated via the formula
$0.1 = (1-\epsilon)^{N(n-1)}$, and at a gate error $\epsilon=10^{-3}$ it limits the maximal number
of qubits as $N(n-1) \leq 2301$ -- which is still an encouraging figure in practice.
We note that other error mitigation schemes could also be applied straightforwardly
to address errors happening during the derangement measurement.

\begin{figure*}[tb]
	\begin{centering}
		\includegraphics[width=0.95\textwidth]{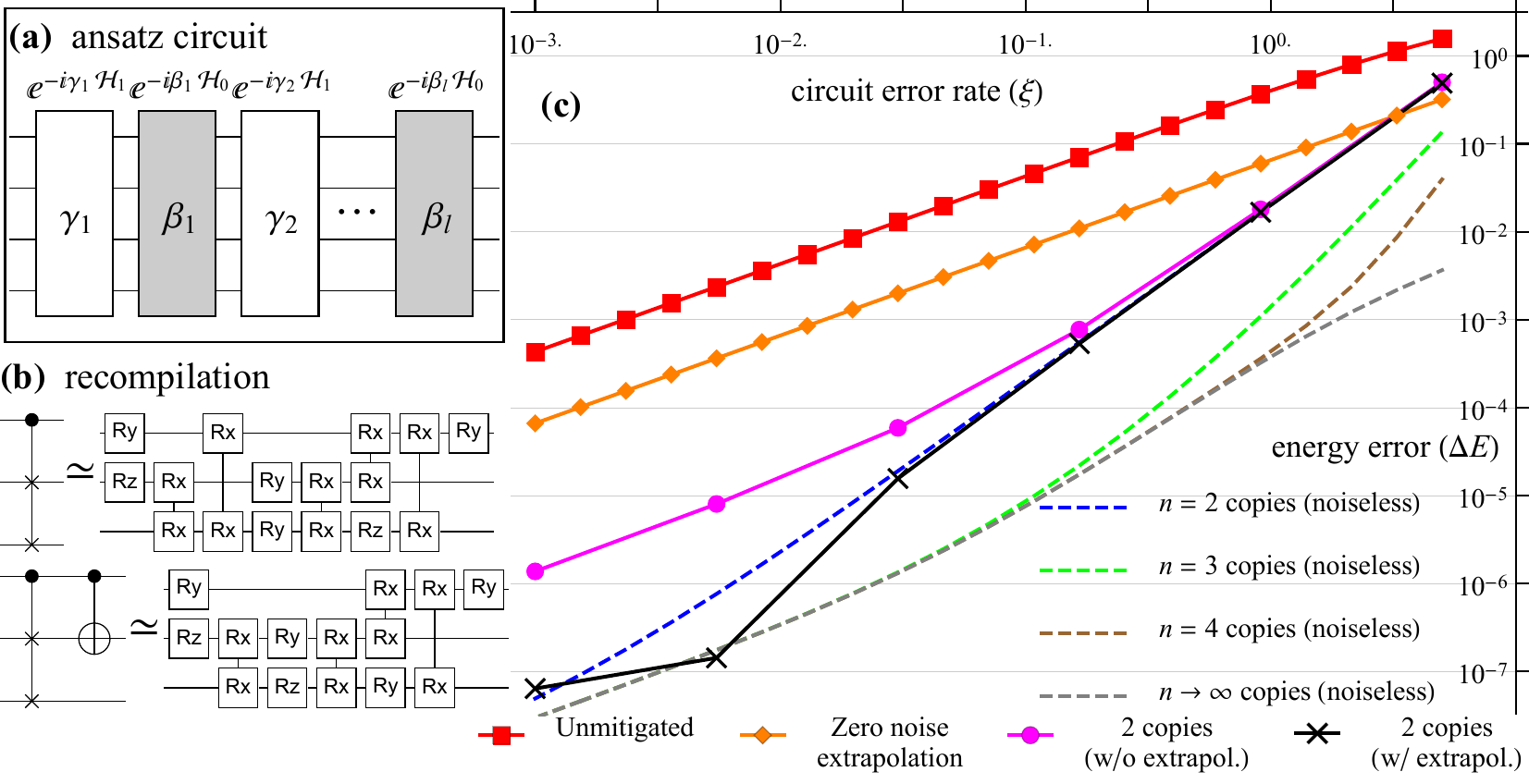}
		\vspace{-4mm}
		\caption{
				(a) Ansatz used to prepare the ground state of the Hamiltonian in Eq.~\eqref{spin_ring}
				for $N=6$ qubits.
				(b) Type B (type C) recompilation of controlled-SWAP (including observable) gates from Table~\ref{recompile}
				requires $5$ ($4$) applications of the hardware-native entangling gates.
				(c) Error in estimating the ground state energy with and without mitigation as a function
				of the number of expected errors $\xi$ in the ansatz circuit. We assume that the experimentalist
				can amplify the vast majority of the noise (94$\%$) in the hardware-native gates, but not all of it, limiting
				extrapolation to a finite precision (orange diamonds).
				Although the derangement circuit is also degraded by noise, it can still drastically reduce errors
				both in combination with (black crosses) and without extrapolation (magenta dots).
				Dashed lines correspond to $\tr[\mathcal{H} \rho^n]/\tr[\rho^n]$ as obtained via noiseless derangement circuits.
				When increasing $n$, we approach in exponential order a non-zero error (dashed grey) which is due to
				the coherent mismatch in the dominant eigenvector.
				The present demonstration on $2\times6+1$ qubits should rather be viewed as a worst-case
				scenario since increasing the scale of the computation will favour the ESD approach.
			\label{groundstate}
		}
	\end{centering}
\end{figure*}

\subsection{Limitations of the Technique}

There is one main limitation of
the present approach: In a realistic experiment one can expect that
coherent errors occur. 
As opposed to error correcting schemes, our ESD approach is completely oblivious to
these and ultimately such errors will limit precision.
Nevertheless, well-established techniques enable us to suppress these coherent errors, e.g.,
via converting them into incoherent errors by Pauli twirling \cite{PhysRevA.78.012347, PhysRevA.85.042311, cai2019constructing, cai2020mitigating}.
Furthermore, as discussed above, even incoherent noise models introduce a mismatch
in the dominant eigenvector which can be expressed via $ \sqrt{1-c} | \psi_{id} \rangle + \sqrt{c} | \psi_{err} \rangle$.
While the coherent mismatch $c$ limits the precision of the present approach,
we present a comprehensive analysis and provide strong theoretical guarantees in ref.~\cite{koczor2021dominant}
that its impact decreases when increasing the scale of the computation.
Refer also to the Appendix for an illustration how this error can be mitigated.

Furthermore, the present approach is expected to be particularly well suited for variational quantum
algorithms: First, the impact of coherent mismatch is
guaranteed to be quadratically smaller when the aim is to prepare eigenstates~\cite{koczor2021dominant}.
Second, variational algorithms are inherently robust to this kind of error as a variational optimisation
implicitly minimises the impact of coherent errors.
We also remark that in the context of variational algorithms one could slightly re-adjust
variational parameters such that the overlaps between copies $\tr[\rho_k \rho_l]$
are maximal for every $k \neq l$ -- note that measuring such overlaps is possible with the
setup in Fig.~\ref{schem}. This ensures us that the dominant eigenvector of every
copy is (close to) identical.
One could also use Clifford circuits to calibrate or validate the quantum device by
comparing to expectation values obtained from (efficient) classical simulations \cite{strikis2020learning, czarnik2020error}.

We further remark that we have also neglected the effect of measurement errors, i.e.,
when the probability of collapsing into state $0$ is biased. Nevertheless, there exist
well-established techniques for mitigating the effect of such imperfections
\cite{maciejewski2020mitigation,endo2020hybrid}.

\section{Practical Applications}

Recall that near-term quantum devices are limited to
shallow quantum circuits due to their inability to implement quantum error correction.
Nevertheless, such shallow circuits may still be of high practical value as, for example,  they may allow one
to approximate ground-state energies of Hamiltonians $\mathcal{H}$, which
cannot be estimated by other means~\cite{cerezo2020variationalreview,endo2020hybrid,bharti2021noisy,
	farhi2014quantum,peruzzo2014variational,wang2015quantum,PRXH2,PhysRevA.95.020501,mcclean2016theory}.
Let us consider a spin-ring Hamiltonian with a constant coupling $J=0.1$ and uniformly
randomly generated on-site interaction strengths $\omega_k \in [-1,1]$ as 
\begin{equation}\label{spin_ring}
	\mathcal{H} =  \sum_{k \in \text{ring}(N)} \omega_k Z_k  + J \, \vec{\sigma}_k \cdot \vec{\sigma}_{k+1},
\end{equation}
for the following reasons:
(a) this Hamiltonian is relevant in the context of condensed matter
phenomena, such as manybody localisation \cite{nandkishore2015many}, but its
ground state cannot be approximated classically for large $N$ \cite{spin_ring_hamil,childs2018toward};
(b) it has a very simple structure as well as a linearly scaling number of Pauli observables $\sigma$;
(c) it is closely related to other important Hamiltonians, cf. approximate optimisation algorithms
(QAOA) or spin systems in materials science~\cite{cerezo2020variationalreview,endo2020hybrid,bharti2021noisy, pagano2020quantum,arute2020quantum}.

We prepare the ground state via the usual variational Hamiltonian ansatz
(VHA)~\cite{cerezo2020variationalreview,endo2020hybrid,bharti2021noisy}, which
was proposed in the context of QAOA~\cite{farhi2014quantum,pagano2020quantum,arute2020quantum},
but has successfully been extended to and analysed in the context of, e.g., quantum chemistry, the Hubbard model
as well as spin systems \cite{babbush2018low,wecker2015progress,cade2020strategies,PRXQuantum.1.020319}.
It consists of alternating layers of discretised time evolutions as illustrated in Fig~\ref{groundstate}/a,
refer to the Appendix for more details.
We consider a quantum device that can natively implement single-qubit $R_y$ and $R_z$ rotation
gates as well as XX gates of the form $\exp[-i\theta X_j X_k]$ between any pairs $j \neq k$ of qubits, i.e.,
a gateset comparable to ion-trap systems \cite{pogorelov2021compact}. Such a platform can efficiently implement
the ansatz circuit of $l$ layers using $3Nl$ applications of the entangling gates.
Using general techniques of ref.~\cite{khatri2019quantumassisted} we
recompile the derangement circuit into hardware-native quantum operations.
Table~\ref{recompile} summarises the number of entangling
($\nu_e$) and single-qubit ($\nu_s$) gates required to implement
the elementary controlled-SWAP operator: we find more compact representations
than previous ones \cite{PhysRevLett.75.748,smolin1996five}.

We use $l=20$ ansatz layers such that the ground state energy in a noise-free setting could be approximated to $\Delta E \approx 10^{-4}$
and explicitly simulate $N=6$ qubits with $n=2$ copies of the noisy computational state (equivalent of a 26-qubit pure-state simulation).
We discuss in the Appendix that controlled-SWAP gates in the derangement circuit need only
be recompiled up to a local $SU(4)$ freedom as shown in Fig.~\ref{groundstate}(b),
refer also to second and third columns in Table~\ref{recompile}.
We thus need less than $5 N = 30$ entangling gates for the mitigation, which
is significantly fewer than the $3Nl = 360$ entangling gates required for the the main computation.
Fig.~\ref{groundstate}(c/red squares) shows unmitigated errors when estimating the ground state energy
of $\mathcal{H}$. We assume a noise model in which the vast majority of errors is due to dephasing
and damping (relaxation), which the experimentalist can perfectly amplify.
We additionally assume that a small depolarising noise,
approximately $6\%$ of the overall gate error rate,
affects the qubits that the experimentalist cannot amplify.
This limits extrapolation techniques \cite{endo2020hybrid,Li2017,PhysRevX.8.031027, kandala2018extending,PhysRevLett.119.180509,cai2020multi}
to a finite precision as shown in Fig.~\ref{groundstate}(c/orange diamonds).
In contrast, the present approach can suppress errors under arbitrary noise models.
Indeed, even with a noisy derangement
circuit, one can drastically reduce errors by orders of magnitude as shown in Fig.~\ref{groundstate}(c/magenta dots).

As discussed above, we can apply zero-noise extrapolation to mitigate the effect of errors in the
derangement circuit. As such, extrapolation in Fig.~\ref{groundstate}(c/black crosses)
can almost fully mitigate errors in the derangement circuit as black crosses approach the blue
dashed line, i.e., the performance of the noiseless derangement circuit $D_2$.
Thus it would be advantageous to prepare a larger number of copies $n>2$ 
to further suppress the errors as illustrated in Fig.~\ref{groundstate}(c/green and brown dashed lines).
We remark, however, that going significantly beyond $n=4$ copies may not be relevant in practice for the following reasons.
(a)
In the practically most important region with $\xi \lessapprox 1$, errors may be sufficiently suppressed
below the level of other practical factors, such as shot noise, or the approximation error $\Delta E \approx 10^{-4}$ due to
insufficient ansatz depth.
(b)
In the limit of a large number of copies, i.e., $n \rightarrow \infty$, a constant error is approached
which is due to the coherent mismatch Fig.~\ref{groundstate}(c/grey dashed line).
(c)
The region with $\xi \gtrapprox 2$ is practically inaccessible due the to rapidly increasing measurement overhead
from Result~\ref{result3} via $f=\mathcal{O}(\ln\lambda^{-1}) = \mathcal{O}(\xi)$
using that $\lambda =\mathcal{O}( e^{-\xi} )$ \cite{koczor2021dominant}.
Note that it is generally the drawback of all mitigation techniques that their measurement cost
grows exponentially with $\xi$ and becomes prohibitive when $\xi \gtrapprox 2 $\cite{endo2020hybrid}.

Let us finally emphasise that one should look at the present demonstration as a worst-case scenario
for the following reasons.
(a) Practical value is expected when computations are scaled
beyond $N>20$ qubits \cite{spin_ring_hamil,childs2018toward}, for which the
ansatz layers need to be increased beyond the present $l=20$,
e.g., refer to \cite{niu2019optimizing,zhou2020quantum,PRXQuantum.1.020319}.
This leads to an increasing ratio $r_{e}$ of the number of entangling gates in the main computation relative
to the derangement circuit as $r_{e} = \frac{3}{5} l(N)$. Here the number of layers $l(N) > \mathcal{O}(N^{0})$
needs to grow faster than a constant. 
(b) The impact of coherent mismatch in Fig.~\ref{groundstate}(c/grey dashed line) is guaranteed to
decrease as the number of gates increases \cite{koczor2021dominant}.
(c) Approximating ground states of Hamiltonians other than the one in Eq.~\ref{spin_ring} may require
more complex ansatz circuits with more rapidly growing gate counts.
For example, simulating the Hubbard model on $N=50$ qubits---one of the promising
candidates for demonstrating practical quantum advantage---requires $\approx 2\times10^4$ entangling
gates~\cite{PhysRevApplied.14.014059}
while the derangement circuit requires only a few hundred, resulting in the ratio of entangling gates as $r_e\approx10^2$.
An even more pronounced example is the case of molecular Hamiltonians in which the number of
Pauli terms may grow as $\mathcal{O}(N^4)$ \cite{ourReview}.
(d) The ansatz was optimised in a noiseless, pure-state simulation and re-optimising the
parameters may reduce the impact of coherent mismatch.
(e) In the present case we assume $94\%$ of gate errors can be amplified perfectly:
the experimentalist may only have control of a smaller fraction of errors further limiting
the precision of extrapolation techniques.

We also consider the example of a connectivity constrained architecture in the Appendix: the number
of two-qubit gates to implement the derangement circuit is increased from $5N$ to $6N$, while in the
ansatz it is increased from $3Nl$ to $9Nl$. Thus in such a
scenario connectivity constraints work in our favour. Of course, in principle specific hardware
may be fabricated to optimally accommodate the present technique as well as one may
utilise long-range links between macroscopically separate quantum processors \cite{PhysRevLett.124.110501}.

\section{Discussion and Conclusion}

This work has introduced a novel principle
for suppressing errors in near-term quantum devices. As opposed to error mitigation
techniques, our ESD approach requires an increased system size:
By preparing $n$ identical copies of a computational state, our derangement
circuit protects its permutation symmetry and suppresses errors in an
expectation value measurement exponentially (in the number $n$ of copies).
Furthermore, the ESD is very NISQ friendly, since the $n$ copies
of the computational state can be prepared completely independently
and they only need to be `bridged' by a shallow derangement circuit
immediately prior to measurement.
Furthermore, the significant advantage of the ESD approach is that it is completely oblivious to the error
model during the state preparation process and works (in principle) with arbitrarily high error
rates.
As such, the present approach could be compared to other mitigation techniques.
While quasi-probability techniques
\cite{PhysRevLett.119.180509,PhysRevX.8.031027, strikis2020learning, czarnik2020error}
may in principle be able to perfectly negate the effect of errors, they require
an exponentially growing number of circuit variants together with a perfect knowledge
of the error model. Any deviation from the assumed noise model results in errors,
which may grow exponentially with the number of gates.
Symmetry verification is another successful mitigation technique that could be used
if exploitable symmetries are present \cite{PhysRevA.98.062339,PhysRevLett.122.180501},
however, it cannot reduce errors that fall within the subspace of appropriate symmetry.
Furthermore, zero-noise extrapolation \cite{endo2020hybrid,Li2017,PhysRevX.8.031027, kandala2018extending,PhysRevLett.119.180509,cai2020multi}
can in principle be applied generally, however, the experimentalist may not
be able to perfectly amplify all errors, see Fig.~\ref{groundstate}(c). 
In contrast, the present approach can be applied completely generally in any scenario.
Note, however, that these existing mitigation techniques will be highly relevant as they can
be used in combination with the present approach, as demonstrated above.

The main limitation of the ESD approach is that it
cannot address coherent noise or a coherent mismatch in the dominant eigenvector,
although those errors can be exponentially smaller than the
incoherent decay of the fidelity and are guaranteed to
decrease when increasing the scale of the computation~\cite{koczor2021dominant}.
As long as the derangement circuit is assumed to be perfect, the sample complexity of
our ESD approach is  polynomial in the inverse precision $\mathcal{E}^{-1}$ and comparable to the standard shot-noise limit in
practically relevant scenarios, i.e., when the number of expected errors in the
main computation is below $\xi \approx 1$. Errors during the derangement process
do degrade the performance of the present approach and one needs to rely
on error mitigation techniques to reduce this impact. 
Nevertheless, it was shown above that the number of gates in the derangement
circuit is expected to become negligible relative to the main computation when scaling up computations.

Let us now briefly comment on prior approaches that similarly consider
identical copies of quantum states and similarly apply SWAP operators (or generalisations thereof).
In fact, numerous prior works have considered and exploited
the permutation symmetry of identical
copies of mixed states in the context of, e.g., reconstructing spectral properties of mixed
quantum states~\cite{ekert2002direct, PhysRevA.64.052311, marvian2014generalization, 9163139, PhysRevA.89.012117,christandl2007nonzero,christandl2006spectra},
probing their entanglement characteristics \cite{PhysRevA.68.052101, PhysRevLett.90.167901, PhysRevLett.89.127902},
for constructing universal quantum software \cite{PhysRevLett.89.190401},
and for optimal state discrimination \cite{PhysRevA.66.022112, PhysRevA.65.062320, PhysRevA.98.062318}.
Indeed, in the special case of $n=2$ copies, our scheme is comparable to a modification
of the usual SWAP-test circuit \cite{ekert2002direct}.
However, as opposed to previous works,
here we are not interested in the input mixed state $\rho$, but only in its dominant eigenvector $|\psi\rangle$
that represents a computational quantum state. In fact, we regard any other contribution 
in the state $\rho$ as `noise' which we aim to exclude from the expectation-value measurement process.
The present approach could also be compared to entanglement distillation protocols
\cite{PhysRevLett.76.722,kalb2017entanglement,deutsch1996quantum}, however, our derangement
circuit cannot exponentially improve the `quality' of the input states, but only exclude erroneous contributions
from the expectation-value measurement process.

Let us finally remark that the ESD approach leaves a lot of room for a large
number of different physical implementations, beyond the
circuit in Fig.\ref{schem} that has been analysed in detail in this work.
Our circuit in Fig.\ref{schem} is only one possible realisation
of the general principle outlined here and even this circuit
has a large number of invariants.	We only need to remark here that the results presented
here are very general, and our example circuit could certainly be improved by
combining it with advanced techniques for example, by simultaneously measuring groups of commuting
observables \cite{Crawford2019,hadfield2020measurements} -- but we expect these can only introduce constant factor improvements
and will not change the main results in this work.
In future work we will explore the numerous possibilities offered by the
general principle introduced here. 

Please also refer to the online repository~\cite{Note\thefootcombined,git_derangement_circuits} for simulation and demonstration material.

\section*{Acknowledgments}
I would like to thank Simon C. Benjamin for his invaluable comments and
challenging questions. His help and support was crucial for finalising
this work. I would like to thank Earl Campbell, Robert Zeier, Suguru Endo and Ying Li
for their very constructive and valuable comments on drafts of this work.
I acknowledge funding received from EU H2020-FETFLAG-03-2018 under Grant Agreement No. 820495
(AQTION) and the QCS Hub (EPSRC Hub grant under the agreement number EP/T001062/1)
for support including hardware provision.
The numerical modelling involved in this study made
use of the Quantum Exact Simulation Toolkit (QuEST), and the recent development
QuESTlink\,\cite{QuESTlink} which permits the user to use Mathematica as the
integrated front end. I am grateful to those who have contributed
to both these valuable tools. 

\vspace{4mm}

\emph{Note on subsequent works}---A week after I had made my
	preprint available a paper appeared on the arXiv that proposes a very similar idea \cite{huggins2020virtual},
	while mostly focusing on the $n=2$ scenario.
	The main difference is that ref.~\cite{huggins2020virtual} uses the trace distance to quantify errors thereby
	also taking into account the effect of coherent mismatch. In contrast, in Result~\ref{result1} I only quantify
	errors with respect to the dominant eigenvector, while I defer a comprehensive analysis of the
	coherent mismatch to the subsequent paper \cite{koczor2021dominant}
	for the following reasons.

	(a) Ref.~\cite{huggins2020virtual} numerically computed and plotted the
	trace distance in a comprehensive range of scenarios, and 
	noted that the bound is ``pessimistic'' as it overestimates errors.
	As such, in ref.~\cite{koczor2021dominant} I show that in most
	practical scenarios this trace distance should not be used	
	since a quadratically smaller bound exists, i.e., the square  $c$ of the trace distance $\sqrt{c}$.
	This is nicely illustrated in Fig.~\ref{groundstate}(c/gray dashed line):
	at a circuit error rate $\xi=0.1$ the actual error is $\Delta E \approx 7.5 \times 10^{-6}$
	while the bound of ref.~\cite{huggins2020virtual} is misleading as it
	is orders of magnitude larger
	$2 \sqrt{c} \lVert \mathcal{H}\rVert_{\infty}= 9.2 \times 10^{-3}$. 
	The relation between the two bounds is discussed in more detail in ref.~\cite{koczor2021dominant}.

	(b) Going beyond the ``pessimistic'' bound of
	ref.~\cite{huggins2020virtual} and realistically characterising 
	the coherent mismatch is a very complex problem
	as it is related to important themes in mathematics, such as Weyl's inequalities
	--  solving of which was a major breakthrough.
	In ref.~\cite{koczor2021dominant} I provide strong theoretical guarantees that the coherent mismatch
	can be exponentially small and decreases when increasing the scale of the computation.
	Thus practitioners need principally care about the errors with respect to measuring expectation values in
	the dominant eigenvector, cf. Fig.~\ref{groundstate}.

	(c) It is also interesting to note that Result~\ref{result1} only depends on spectral properties, i.e.,
	eigenvalues $\lambda_k$ and R{\'e}nyi entropies $H_n$,
	that may be estimated experimentally. In contrast, estimating the trace distance of ref.~\cite{huggins2020virtual}
	would require one to prepare
	the ideal, perfect, noiseless quantum state as well as the state $\rho^n/\tr[\rho^n]$, which is prohibitive.

In the interval since the present paper and then ref.~\cite{huggins2020virtual} appeared,
other studies have already reported ideas extending or varying these original concepts.
For example, ref.~\cite{cai2021quantum} introduces a generalisation of the presented
permutation-symmetry principles. Furthermore, ref.~\cite{czarnik2021qubit} proposes that the derangement circuit
$D_n := \mathrm{SWAP}_{1,n} \cdots \, \mathrm{SWAP}_{1,3} \, \mathrm{SWAP}_{1,2}$ can be realised in a
qubit-efficient manner by utilising qubit resets, thus drastically reducing resource requirements
of the present approach.


%

\onecolumngrid
\appendix
\clearpage

\section{Derangement measurements and suppressing errors}

In this section we prove that the derangement circuit in Fig.~\ref{schem}
can be used to estimate expectation values. We then prove upper bounds
on approximation errors with or without using R{\'e}nyi entropies of quantum states.
We finally prove sample complexities of our ESD approach.

\begin{definition}\label{derangement_def}
	We define the set $\mathfrak{D}_n$ of derangement operators that permute $n\geq 2$ quantum registers via
	their unitary representation as
	\begin{equation*}
	\text{every}  \quad \perm \in \mathfrak{D}_n,  \quad \text{is such that}
	\quad	\perm | \psi_1 , \psi_2 , \dots \psi_n \rangle = |\psi_{s(1)} , \psi_{s(2)}  , \dots \psi_{s(n)} \rangle.
	\end{equation*}
	Here all $s \in S_n$  are permutations of the index set $\{1,2,\dots n\}$
	with no fixed point, i.e., $s$ are derangements \cite{sagan2013symmetric,roberts2009applied}.
	Here $S_n$ denotes the symmetric group.
	For $n \geq 4$ we also
	demand that $s$ are $n$-cycles (standard cyclic permutations of maximal length \cite{sagan2013symmetric}), which are a subset of derangements.
	The number of unique ($n$-cycle) derangement operators is given as $|\mathfrak{D}_n| = (n-1)!$.
	Due to seminal results of D{\'e}nes, $s$ can be decomposed into $n-1$ transpositions \cite{denes1959representation}
	and therefore $\perm$ decomposes into $n-1$ pair-wise SWAP operators of the quantum registers.
	One can therefore construct minimal SWAP circuits by (bijectively) mapping the corresponding permutations
	performed by $\perm \in \mathfrak{D}_n$ to graph trees.
\end{definition}

\begin{theorem}	\label{theo1}
	We consider $n$ identical copies of the same quantum register $\rho$ in a separable state as $\rho^{\otimes n}$.
	Methods A and B, as illustrated in Fig.~\ref{schem}, result in the probability of
	measuring the ancilla in the 0 state as
	\begin{equation}
	\text{Method A/B:} \quad \quad \mathrm{prob}_0 = \tfrac{1}{2} + \tfrac{1}{2} \tr[\rho^n \sigma].
	\end{equation}
	Here $\sigma$ is a unitary (Hermitian) observable (or otherwise the real part of a unitary operator is estimated).
\end{theorem}

\begin{proof}
	We start by recapitulating that any density operator $\rho$ admits the following spectral decomposition 
	(note that here we use a different notation than what in the main text)
	\begin{equation}
	\rho = \sum_{k=1}^{2^N} p_k |\psi_k \rangle \langle \psi_k |,
	\quad \quad \text{thus} \quad \quad
	\rho^{\otimes n} = \sum_{k_1,k_2,\dots k_n=1}^{2^N} p_{k_1}p_{k_2}\cdots p_{k_n}  
	|\psi_{k_1}, \psi_{k_2},  \dots \psi_{k_n} \rangle \langle \psi_{k_1}, \psi_{k_2},  \dots \psi_{k_n} |,
	\end{equation}
	where the second equation is the spectral decomposition of $n$ copies of the same state.

	Recall that the action of any unitary circuit $U$ on a density matrix
	$U \rho U^\dagger$ represents a probabilistic mixture
	of its transformed eigenvectors $U |\psi_k \rangle$ that occur with probabilities $p_k$.
	Similarly the action of a unitary circuit on the composite state $\rho^{\otimes n}$ can
	be written as a probabilistic mixture of the pure-states $U |\psi_{k_1}, \psi_{k_2},  \dots \psi_{k_n} \rangle$
	that occur with probabilities as products $p_{k_1}p_{k_2}\cdots p_{k_n}$.
	
	Let us now derive the action of the unitary circuit in Fig.~\ref{schem} on the composite quantum system
	$\rho^{\otimes n}$.
	Our proof works with any derangement operator $\perm$ from Definition~\ref{derangement_def}
	but here we only need to consider one example: we consider a cyclic shift
	(as originally proposed in \cite{ekert2002direct}) of the registers 
	via its explicit action on pure states as
	$$  \perm | \psi_1 , \psi_2 , \dots \psi_n \rangle = |\psi_n , \psi_1  , \dots \psi_{n-1}\rangle. $$

	Our controlled derangement operator acts on the pure state $|0,\psi_{k_1}, \psi_{k_2},  \dots \psi_{k_n} \rangle$
	that occurs with a probability $p_{k_1}p_{k_2}\cdots p_{k_n}$, and we denote as $0$ the state of the additional ancilla qubit.
	Applying the sequence of gates from Fig.~\ref{schem} yields the following transformations of the pure states.
\begin{gather*}
|0,\psi_{k_1}, \psi_{k_2},  \dots \psi_{k_n} \rangle\\
 \Big\downarrow \mathrm{H}\\
\left( |1,\psi_{k_1}, \psi_{k_2},  \dots \psi_{k_n} \rangle + |0,\psi_{k_1}, \psi_{k_2},  \dots \psi_{k_n} \rangle \right)/\sqrt{2}\\
 \Big\downarrow \mathrm{controlled} \, \perm\\
\left(|1,\psi_{k_n}, \psi_{k_1},  \dots \psi_{k_{n-1}} \rangle + |0,\psi_{k_1}, \psi_{k_2},  \dots \psi_{k_n} \rangle \right)/\sqrt{2}\\
 \Big\downarrow \mathrm{controlled} \, \sigma\\
\left(|1,\sigma \psi_{k_n}, \psi_{k_1},  \dots \psi_{k_{n-1}} \rangle + |0,\psi_{k_1}, \psi_{k_2},  \dots \psi_{k_n} \rangle \right)/\sqrt{2}\\
  \Big\downarrow \mathrm{H}\\
 \left(|0,\sigma \psi_{k_n}, \psi_{k_1},  \dots \psi_{k_{n-1}} \rangle + |0,\psi_{k_1}, \psi_{k_2},  \dots \psi_{k_n} \rangle \right)/2
 + \dots.
\end{gather*}
	It is now straightforward to show that the probability of
	measuring the ancilla qubit in state 0 is
	\begin{equation}
	\mathrm{prob}_0 = \frac{1}{2} + \frac{1}{2} \sum_{k_1,k_2,\dots k_n=1}^{2^N} p_{k_1}p_{k_2}\cdots p_{k_n} 
	\langle \psi_{k_1}, \psi_{k_2},  \dots \psi_{k_n}	| \sigma \psi_{k_n}, \psi_{k_1},  \dots \psi_{k_{n-1}} \rangle,
	\end{equation}
	where we can simplify the inner products as
	\begin{equation} \label{orthogonality}
	\langle \psi_{k_1}, \psi_{k_2},  \dots \psi_{k_n}	| \sigma \psi_{k_n}, \psi_{k_1},  \dots \psi_{k_{n-1}} \rangle
	=\langle\psi_{k_1} |\sigma \psi_{k_n} \rangle \langle  \psi_{k_2} | \psi_{k_1} \rangle \cdots \langle \psi_{k_n} | \psi_{k_{n-1}} \rangle
	=\langle\psi_{k_1} |\sigma \psi_{k_n} \rangle \delta_{{k_2} {k_1}} \cdots \delta_{{k_n} {k_{n-1}}},
	\end{equation}
	and we have used the orthogonality of the eigenstates $| \psi_{k} \rangle$ and $\delta_{a b}$ is
	the Kroenecker delta symbol.
	At this point we remark that our proof works with any derangement operator from Definition~\ref{derangement_def}
	since these will conserve the above orthonormality relation. We remark here that the corresponding permutations
	can be mapped to graph trees, which related to the pairs of indexes $\delta_{a b}$ in the
	Kroenecker delta symbols in the above equation.

	Substituting the above results back we obtain the expression for
	the ancilla probability by using that only terms with coinciding indexes contribute to the sum via $k_1 = k_2 =\dots k_n$
	\begin{equation}
	\mathrm{prob}_0 = \frac{1}{2} + \frac{1}{2} \sum_{k=1}^{2^N} p_{k}^n \langle\psi_{k} |\sigma| \psi_{k} \rangle
	=\frac{1}{2} + \frac{1}{2}  \tr[\rho^n \sigma].
	\end{equation}

\end{proof}

\begin{example} \label{example1}
	\normalfont
	Using our definition of
	experimental quantum states from Eq.~\eqref{rhodef}, our circuit in
	Fig.~\ref{schem} can estimate the expectation value
	\begin{equation*}
	\tr[\rho^n \sigma] = \lambda^n \langle\psi |\sigma | \psi \rangle
	+  (1-\lambda)^n\sum_{k=2}^{2^N} p_k^n  \langle\psi_k |\sigma | \psi_k \rangle.
	\end{equation*}
	It is clear that the error probabilities $p_k$ are suppressed exponentially
	via $p_k^n$, but the dominant term gets slightly attenuated too via $\lambda^n$.
	For example, let us assume that our dominant eigenvalue is $\lambda=0.8$ and we have
	a high-entropy error in a subspace spanned by $100$ eigenvectors
	via the uniform distribution $p_k=(1-\lambda)/100$ when $k\leq101$ and $p_k = 0$
	when $k>101$. We then obtain the estimate
	\begin{equation*}
	\tr[\rho^n \sigma] = 0.8^n \langle\psi |\sigma | \psi \rangle
	+  \sum_{k=2}^{101} (2 \times 10^{-3})^n  \langle\psi_k |\sigma | \psi_k \rangle.
	\end{equation*}
	Since $|\langle\psi_k |\sigma | \psi_k \rangle| \leq 1 $, we can upper bound 
	the errors in $\tr[\rho^n \sigma] = 0.512 \langle\psi |\sigma | \psi \rangle + E$
	for, e.g., $n=3$ as $|E|\leq 8 \times 10^{-7}$. Hence our error contribution is at least $640000$-times
	smaller than the desired expectation value. This high degree of error suppression
	is due to the large $n=3$ R{\'e}nyi entropy of the error probabilities $p_k$ as 
	$$H_3(\underline{p}) = - \frac{1}{2} \, \ln [\sum_{k=2}^{2^N} p^n_k  ] = - \frac{1}{2} \, \ln [\sum_{k=2}^{101} (10^{-2})^3  ]
	\approx - \frac{1}{	2} \, \ln [ 10^{-4} ] \approx 4.6.$$
	We will show in Theorem~\ref{theo2} and in Lemma~\ref{renyilemma} that
	the efficiency of the error suppression depends exponentially on this R{\'e}nyi entropy.
	
	Indeed, in order to obtain an accurate estimate of $\langle\psi |\sigma | \psi \rangle$
	we need to have a good knowledge of the largest eigenvalue of the density
	matrix $\lambda$ that divides $\langle\psi |\sigma | \psi \rangle$.
	We assume in Method B in Theorem~\ref{theo2} that this eigenvalue is known precisely.
	However, in Method A we just replace our observable $\sigma$ with the identity in 
	Fig.~\ref{schem} and we directly approximate the $n^{th}$ power of the dominant eigenvalue
	$\lambda$ as 
	\begin{equation*}
	\tr[\rho^n] = 0.8^n +  \sum_{k=2}^{101} (2 \times 10^{-3})^n,
	\end{equation*}
	for $n=3$ we obtain the result as $0.8^3 + 8 \times 10^{-7} = 0.512001$,
	which is a very good estimate of $0.8^3 = 0.512$ as the error is $640000$-times
	smaller than the ideal value.
\end{example}

\begin{example}\label{example2}
	\normalfont We consider now the worst-case scenario of $0$-entropy error distributions.
	For example, let us consider the state $\rho = \lambda | \psi \rangle \langle  \psi | + (1-\lambda) | \psi_{err} \rangle \langle  \psi_{err} |$
	which is a mixture of the ideal state $ | \psi \rangle $ that occurs with a probability $\lambda$
	and an erroneous state $| \psi_{err} \rangle$ which occurs with a probability $(1-\lambda)$.
	The error probability distribution from Eq.~\eqref{rhodef} is obtained as $p_2 = 1$ and $p_k = 0$
	for $k>2$. It follows that the error distribution has a $0$ entropy and our approach
	completely breaks down when $\lambda \leq 1/2$ since the dominant eigenvector then becomes
	$| \psi_{err} \rangle$. We can show that the errors in the expectation value are still
	exponentially suppressed, but much less efficiently than before in Example~\ref{example1}.
	Let us set $\lambda=0.8$ and 
	\begin{equation*}
	\tr[\rho^n \sigma] = 0.8^n \langle\psi |\sigma | \psi \rangle
	+ 0.2^n  \langle\psi_{err} |\sigma | \psi_{err} \rangle.
	\end{equation*}
	For $n=3$ we obtain $\tr[\rho^3 \sigma] = 0.512 \langle\psi |\sigma | \psi \rangle + 0.008 \langle\psi_{err} |\sigma | \psi_{err} \rangle$,
	and therefore the error is suppressed by a factor of $64$. This is significantly
	lower that the factor of $640000$ suppression from Example~\ref{example1} which 
	assumed a high-entropy error distribution.	 
\end{example}

\clearpage

\section{Exponentially decreasing upper bounds on approximation errors}

\begin{theorem}	\label{theo2}
	We use Methods A/B from Fig.~\ref{schem} to estimate the probability
	$\mathrm{prob}_0 = \tfrac{1}{2} + \tfrac{1}{2}\tr[\rho^n \sigma]$ of the ancilla qubit.
	In Method A we use the same technique via $\sigma = \mathrm{Id}$ to estimate the
	probability $\mathrm{prob'}_0 = \tfrac{1}{2} + \tfrac{1}{2}\tr[\rho^n]$
	and our Method A yields the approximation
	\begin{equation*}
	\text{Method A:} \quad \quad 
	\frac{2\mathrm{prob}_0-1}{2\mathrm{prob'}_0-1}
	=
	\frac{\tr[\rho^n \sigma]} {\tr[\rho^n]} =  \langle\psi |\sigma | \psi \rangle + \mathcal{E}_A.
	\end{equation*}
	In Method B we assume that the largest eigenvalue $\lambda$ of the state
	$\rho$ is known and therefore we have the approximation
	\begin{equation*}
	\text{Method B:} \quad \quad 
	(2\prob-1)/\lambda^n
	=
	\tr[\rho^n \sigma] / \lambda^n =  \langle\psi |\sigma | \psi \rangle + \mathcal{E}_B.
	\end{equation*}
	The approximation errors are bounded via
	$|\mathcal{E}_A| \leq \frac{ 2 Q_n}{1+Q_n}$ and $|\mathcal{E}_B| \leq Q_n$,
	and we prove in Lemma~\ref{renyilemma} that the bounding sequence $\qn = (\lambda^{-1}-1)^n \lVert \underline{p} \rVert^n_n$
	generally decays exponentially when we increase $n$ or when we increase the R{\'e}nyi entropy
	of the probability vector $\underline{p}$.
\end{theorem}
\begin{proof}
	Let us recapitulate the explicit form of the density matrix from Eq.~\eqref{rhodef} as
	\begin{equation}
	\rho = \lambda | \psi \rangle \langle \psi  | + (1-\lambda) \sum_{k=2}^{2^N} p_k | \psi_k \rangle \langle \psi_k  |.
	\end{equation}
	We can evaluate the expressions for the trace operation as
	\begin{align*}
	\tr[\rho^n \sigma] =  \lambda^n \langle\psi |\sigma | \psi \rangle + (1-\lambda)^n \sum_{k=2}^{2^N} p^n_k \langle\psi_k |\sigma | \psi_k \rangle,
		\quad \quad \text{and} 	\quad \quad 
	\lvert \sum_{k=2}^{2^N} p^n_k \langle\psi_k |\sigma | \psi_k \rangle \rvert
	\leq  \sum_{k=2}^{2^N} p^n_k 	 = \lVert \underline{p} \rVert^n_n,
	\end{align*}
	where we have used that $|\langle\psi_k |\sigma | \psi_k \rangle| \leq 1$ due to unitarity of $\sigma$
	and $\lVert \underline{p} \rVert_n$ is the $n$-norm of the probability vector $\underline{p}$.

	\textbf{Method B}: Here our aim is to estimate $\tr[\rho^n \sigma] $ and $\lambda^n$ is known exactly.
	The error term can be calculated via
	\begin{equation}
	\lvert \mathcal{E}_B \rvert
	=
	\lvert \frac{\tr[\rho^n \sigma]}{\lambda^n} -  \langle\psi |\sigma | \psi \rangle \rvert
	=
	 \frac{(1-\lambda)^n}{\lambda^n}  \lvert \sum_{k=2}^{2^N} p^n_k \langle\psi_k |\sigma | \psi_k \rangle \rvert
	 \leq \frac{(1-\lambda)^n}{\lambda^n} \lVert \underline{p} \rVert^n_n =: \qn,
	\end{equation}
	and here we have defined the sequence $Q_n$.
	
	\textbf{Method A}: In this case we estimate $\tr[\rho^n \sigma] $  and $\tr[\rho^n]$, and we now calculate the error term
	using that $\tr[\rho^n] = \lambda^n + (1-\lambda)^n \sum_k p_k^n = \lambda^n + (1-\lambda)^n	 \lVert \underline{p} \rVert^n_n$. Indeed, we obtain
	\begin{equation}
	\lvert \mathcal{E}_A \rvert
	=
	\lvert \frac{\tr[\rho^n \sigma]}{\tr[\rho^n]} -  \langle\psi |\sigma | \psi \rangle \rvert
	=
	\lvert \frac{
		 \langle\psi |\sigma | \psi \rangle + \lambda^{-n}(1-\lambda)^n \sum_{k=2}^{2^N} p^n_k \langle\psi_k |\sigma | \psi_k \rangle
	}{
	1 + \lambda^{-n} (1-\lambda)^n  \lVert \underline{p} \rVert^n_n 
	}
	-  \langle\psi |\sigma | \psi \rangle \rvert
	= \lvert \frac{\langle\psi |\sigma | \psi \rangle  + Z}{1 + \qn} - \langle\psi |\sigma | \psi \rangle \rvert,
	\end{equation}
	where we have used the notation $Z = \lambda^{-n}(1-\lambda)^n \sum_{k=2}^{2^N} p^n_k \langle\psi_k |\sigma | \psi_k \rangle$
	for simplicity.
	It follows that the error term is bounded via
	\begin{equation}
	\lvert \mathcal{E}_A \rvert = \lvert \frac{\langle\psi |\sigma | \psi \rangle + Z}{1 + \qn} - \langle\psi |\sigma | \psi \rangle \rvert
	=
	\lvert \frac{ Z - \qn \langle\psi |\sigma | \psi \rangle }{1 + \qn}  \rvert.
	\end{equation}
	Let us now upper bound this expression as
	\begin{equation}	
	 (1+\qn)^{-1} \lvert Z - \qn \langle\psi |\sigma | \psi \rangle \rvert 
	\leq (1+\qn)^{-1} [ \lvert Z \rvert  + \qn \lvert \langle\psi |\sigma | \psi \rangle \rvert],
	\end{equation}
	where we have used the triangle inequality as $|a-b| \leq |a|+|b|$.
	We can now use from before that $\lvert Z \rvert \leq \qn$, which results in the error term
	\begin{equation}
	\lvert \mathcal{E}_A \rvert \leq 
	\frac{\qn+\qn \lvert \langle\psi |\sigma | \psi \rangle \rvert}{1+\qn}
	= \frac{\qn}{1+\qn} (1+\lvert \langle\psi |\sigma | \psi \rangle \rvert)
	\leq\frac{2\qn}{1+\qn}.
	\end{equation}
	This concludes our proof.
	
\end{proof}

\begin{lemma}\label{renyilemma}
	The sequence $Q_n$ in our upper bounds in Theorem~\ref{theo2} decreases exponentially for a fixed $n$ when we increase the
	R{\'e}nyi entropy $H_n(\underline{p})$ of the error probability vector $\underline{p}$ from Eq.~\eqref{rhodef} as
	$\qn =  (\lambda^{-1}{-}1)^n \exp[ - (n{-}1) H_n(\underline{p}) ]$.
	Furthermore, the sequence generally decays exponentially via $\qn \leq  (p_{max})^{{-}1} Q^n$
	where we define the suppression factor $Q := (\lambda^{-1}{-}1) p_{max}<1$ and  $p_{max}$ is the largest error probability from Eq.~\eqref{rhodef}.
\end{lemma}
\begin{proof}
	The first part of the proof straightforwardly follows by substituting the expression for the R{\'e}nyi entropy \cite{renyi1961measures}
	\begin{equation}
	H_n(\underline{p}) = \frac{1}{1-n} \, \ln [\sum_{k=2}^{2^N} p^n_k  ]
	= \frac{n}{1-n} \, \ln \lVert \underline{p} \rVert_n,
	\end{equation}
	into the expression for $Q$ as
	\begin{equation*}
	\qn = [ (\lambda^{-1}-1) \lVert \underline{p} \rVert_n ]^n=   (\lambda^{-1}-1)^n \exp[ - \tfrac{n-1}{n} H_n(\underline{p}) ]^n
	= (\lambda^{-1}{-}1)^n \exp[ - (n{-}1) H_n(\underline{p}) ].
	\end{equation*}
	This concludes the first part of our proof.
	
	Let us now prove that the sequence $\qn$ decreases in exponential order when we increase $n$.
	Using the well-known
	series of inequalities satisfied by the R{\'e}nyi entropies
	as $H_\infty(\underline{p}) \dots \leq  H_n(\underline{p}) \leq  H_{n-1}(\underline{p}) \leq \dots \leq  H_1(\underline{p})$,
	we obtain the general bound $-\ln p_{max} \leq  H_n(\underline{p})$ for all $n$,
	and we define the largest error probability $p_{max}:= \max_k p_k$.
	It follows that	
	$$\qn \leq  (\lambda^{-1}{-}1)^n (p_{max})^{n{-}1} =  (p_{max})^{{-}1} [(\lambda^{-1}{-}1) p_{max} ]^n =: (p_{max})^{{-}1} Q^n .$$
	
		The upper bound $Q<1$ holds due to our condition below Eq.~\eqref{rhodef}
		as $(\lambda-1) p_k < \lambda$ for every probability $k=\{2,3, \dots 2^N\}$.
		It follows that $p_{max} <  (\lambda^{-1}{-}1)^{-1}$ and therefore $Q = (\lambda^{-1}{-}1) p_{max}  <1 $.

\end{proof}

\begin{lemma}\label{complexity_lemma}
	Determining the expectation value $ \langle\psi |\sigma | \psi \rangle$ from Theorem~\ref{theo2} to a
	fixed precision $\mathcal{E}$ requires $n=\frac{\ln \mathcal{E}^{-1} +  \ln{[2    (p_{max})^{{-}1}   ]}   }{\ln Q^{-1}}$ 
	copies of the quantum state $\rho$ (one needs to apply the ceiling function to round this up 
		to the nearest integer). Here $Q<1$ is the
	suppression factor from Theorem~\ref{theo2} and from Lemma~\ref{renyilemma}.
	
	Reducing shot noise to the desired precision $\mathcal{E}$ requires the following number of samples.
	In Method A one needs to assign $N_{s,1}$ samples to determine $\prob$
	and $N_{s,2}$ samples to determine $\mathrm{prob'}_0$ as
	\begin{equation}
	\text{Method A:} \quad \quad	N_{s,1} = \mathcal{O}[ \mathcal{E}^{-2 (1+f) } ] = \mathrm{poly}(\mathcal{E}^{-1})
	\quad \quad \text{and} \quad \quad
	N_{s,2} = \mathcal{O}[ \mathcal{E}^{-2 (1+2f) } ] = \mathrm{poly}(\mathcal{E}^{-1}).
	\end{equation}
	The overall number of measurements required is $N_{s} = N_{s,1}+N_{s,2} = \mathcal{O}[ \mathcal{E}^{-2 (1+2f) } ]$.
	In Method B one only needs to determine $\prob$ since $\lambda^n$ is known. The
	number of samples scales as
	\begin{equation}
	\text{Method B:} \quad \quad N_s  =\mathcal{O}[ \mathcal{E}^{-2 (1+f) } ] = \mathrm{poly}(\mathcal{E}^{-1}).
	\end{equation}
	Indeed, in both cases the measurement cost grows polynomially with the inverse precision $\mathcal{E}^{-1}$ and
	its polynomial order is determined by $f:= \frac{\ln(\lambda^{-1})}{\ln (Q^{-1})}$.
	The standard shot-noise limit $\mathcal{E}^{-2 }$ is only
	slightly modified by $f$ in the case of good quality quantum states or in the case of high-entropy
	probabilities.
\end{lemma}
\begin{proof}
	Let us first compute the upper bound on the number of copies $n$
	required to achieve a fixed precision $\mathcal{E}\ll1$. We use the upper bounds from Theorem~\ref{theo2}
	as $|\mathcal{E}_A| \leq \mathcal{E} = \frac{ 2 \qn}{1+\qn}\approx 2\qn$ and $|\mathcal{E}_B| \leq \mathcal{E} =  \qn$.
	It is clear that the precision of Method A differs by a factor of $2$ for $\mathcal{E}\ll1$, 
	and we will use this expression for both methods for simplicity.
	Let us use the exponentially decreasing upper
	bounds on $\qn$ from Lemma~\ref{renyilemma} and write $|\mathcal{E}_A| \leq 2  (p_{max})^{{-}1} Q^n $,
	and $|\mathcal{E}_B| \leq 2  (p_{max})^{{-}1} Q^n $,
	where we have defined the suppression factor as $Q := (\lambda^{-1}{-}1) p_{max}<1$.
	It is straightforward to express $n$ as
	\begin{equation}
	n = \frac{\ln \mathcal{E}^{-1} + \ln{[2  (p_{max})^{{-}1}   ]}}{\ln Q^{-1}}.
	\end{equation}

	\textbf{Remark:} Let us further expand the above equation by using our expression
	from Lemma~\ref{renyilemma} as $\qn =  (\lambda^{-1}{-}1)^n \exp[ - (n{-}1) H_n(\underline{p}) ]$,
	which results in
	\begin{equation*}
	- \ln(\mathcal{E}^{-1}) =  \ln \mathcal{E} =  \ln 2 \qn = \ln(2) + n \ln[ (\lambda^{-1}{-}1) ]  - (n{-}1) H_n(\underline{p})
	= \ln(2) + H_n(\underline{p}) + n \{ \ln[ (\lambda^{-1}{-}1) ] - H_n(\underline{p}) \}.
	\end{equation*}
	We can express $n$ as
	\begin{equation*}
	 n=
	\frac{ \ln(\mathcal{E}^{-1}) + H_n(\underline{p})+ \ln(2)}
	{H_n(\underline{p}) - \ln[ (\lambda^{-1}{-}1) ]  } = \mathcal{O}[\ln(\mathcal{E}^{-1}) ].
	\end{equation*}	
	We remark that the denominator is positive due to the
	bound on R{\'e}nyi entropies from Lemma~\ref{renyilemma}
	as  $\ln[ (\lambda^{-1}{-}1) ] < H_n(\underline{p})$.
	One should actually use the ceil function to round up the right-hand
	expression to the nearest integer. Note that the above expression
	implicitly depends on $n$ via the R{\'e}nyi entropy $H_n(\underline{p})$,
	but one could always use the series of inequalities $0 \leq H_n(\underline{p}) \leq H_{n-1}(\underline{p})  \leq \dots H_{2}(\underline{p}) \leq  H_{1}(\underline{p}) $
	to bound the value of $n$. It is straightforward to show now that
	in the limiting scenarios $ H_{2}(\underline{p}) \gg 1$ or $\lambda \approx 1$ we recover $n \rightarrow 1$ (via the ceil function).
	Let us now express the scaling with respect to shot noise. 
	
	\textbf{Method B:} We estimate the probability $\prob$ from Theorem~\ref{theo2} and we exactly
	know $\lambda^n$. Our precision $\mathcal{E}$ is determined by the variance of our estimator which 
	can be obtained as
	\begin{equation}
	\mathcal{E}^2 = \var[ \langle\psi |\sigma | \psi \rangle]
	=
	\var[\frac{2\prob-1}{\lambda^n}]
	= \frac{4\var[\prob]}{\lambda^{2n}}
	= \frac{4\prob(1-\prob)}{N_s \lambda^{2n}},
	\end{equation}
	where we have used that the variance of the binomial distribution is $\prob(1-\prob)/N_s$
	and $N_s$ is the number of samples.
	We can explicitly express the number of shots $N_s$ required to reach a fixed precision
	$\mathcal{E}$ as
		\begin{equation}
	N_s 
	= \frac{4\prob(1-\prob)}{\mathcal{E}^2  \lambda^{2n}},
	\end{equation}
	Let us now simplify $\lambda^{2n}$ by
	expressing the dependence of $n$ on the precision above $\mathcal{E}$ as
	\begin{equation*}
	\ln [ \lambda^{2n} ] = 2n \ln[\lambda]
	=
	2\ln[\lambda]
	 \frac{\ln \mathcal{E}^{-1} + \ln{[2   (p_{max})^{{-}1}  ]}   }{\ln Q^{-1}}
	 =
	 	 \ln \mathcal{E}^{-1} \frac{2 \ln[\lambda]}{\ln Q^{-1}} + \frac{\ln[\lambda]  \ln{[4    (p_{max})^{{-}2}  ]}  }{\ln Q^{-1}},
	\end{equation*}
	and it follows that
		\begin{equation} \label{lambdapower}
	\lambda^{2n}
	=
	\exp[
	\ln \mathcal{E}^{-1} \frac{2 \ln[\lambda]}{\ln Q^{-1}} + \frac{\ln[\lambda]   \ln{[4   (p_{max})^{{-}2}   ]}    }{\ln Q^{-1}}]
	=
	\mathcal{E}^{ \frac{2 \ln[\lambda^{-1}]}{\ln Q^{-1}} } \exp[\frac{\ln[\lambda]   \ln{[4    (p_{max})^{{-}2}   ]}   }{\ln Q^{-1}}].
	\end{equation}
	We can finally express the number of samples explicitly as
		\begin{equation}
		N_s 
		= 4\prob(1-\prob) \mathcal{E}^{-2 [1+ \frac{\ln(\lambda^{-1})}{\ln (Q^{-1})}] } \exp[\frac{\ln (\lambda^{-1})   \ln{[4    (p_{max})^{{-}2}  ]}     }{\ln (Q^{-1})}]
		=\mathcal{O}[ \mathcal{E}^{-2 (1+f) } ]
		\end{equation}
	Here we used that $4 \exp[\frac{\ln[\lambda^{-1}]   \ln{[4   (p_{max})^{{-}2}   ]}     }{\ln (Q^{-1})}]$ is a constant multiplication factor
	and $0 \leq \prob \leq 1$ and we have introduced $f := \frac{\ln(\lambda^{-1})}{\ln (Q^{-1})}$.
	Indeed, we obtain the expected limits due to  $\lim_{\lambda\rightarrow1} f = 0$ and $\lim_{Q\rightarrow0} f= 0$.
	
	In general when $\lambda>1/2$ we can use the expression $Q\leq(\lambda^{-1}-1)$ from Theorem~\ref{theo2} as
	$f \leq \frac{\ln(\lambda^{-1})}{\ln[(\lambda^{-1} -1)^{-1}] }$ which is only saturated by $0$-entropy distributions.
	For example when $\lambda=0.6$ then we obtain $f \leq 1.26$, and this value can be smaller depending on
	the entropy of the probability distribution. Interestingly, for sufficiently good quality states as $\lambda \geq 0.9$,
	the polynomial overhead introduced is very small via $f \leq 0.16$.
	
	\textbf{Method A:} In this case we estimate both $\prob$ and $\mathrm{prob'}_0$. The variance of our estimator can be specified
	as
	\begin{equation*}
	\mathcal{E}^2=
	\var[\langle\psi |\sigma | \psi \rangle ]
		=
	\var[\frac{2\mathrm{prob}_0-1}{2\mathrm{prob'}_0-1}]
	= \var[\mathrm{prob}_0]  \frac{4}{(2\mathrm{prob'}_0-1)^{2}}
	+ \var[\mathrm{prob'}_0]  \frac{ 4 (2\mathrm{prob}_0-1)^2}{ (2\mathrm{prob'}_0-1)^{4} }.
	\end{equation*}
	Let us now use that $2\mathrm{prob'}_0-1 \approx \lambda^n$ and simplify the above expression
	as
	\begin{equation*}
	\mathcal{E}^2=
	\var[\langle\psi |\sigma | \psi \rangle ]
	\approx 
	\var[\frac{2\mathrm{prob}_0-1}{2\mathrm{prob'}_0-1}]
	= \var[\mathrm{prob}_0]  \frac{4}{\lambda^{2n}}
	+ \var[\mathrm{prob'}_0]  \frac{ 4 (2\mathrm{prob}_0-1)^2}{ \lambda^{4n} }.
	\end{equation*}
	We can again substitute the variance of binomial distributions as
	$\var[\mathrm{prob}_0] = \mathrm{prob}_0(1-\mathrm{prob}_0)/N_{s,1}$
	and $\var[\mathrm{prob'}_0] = \mathrm{prob'}_0(1-\mathrm{prob'}_0)/N_{s,2}$.
	The measurement cost of determining both components to a precision
	$\mathcal{E}^2/2$ follows as
	\begin{equation*}
	N_{s,1} = \frac{8\mathrm{prob}_0(1-\mathrm{prob}_0)}{\mathcal{E}^2 \lambda^{2n}}
	\quad \quad \text{and} \quad \quad
	N_{s,2} = 	\frac{ 8 \mathrm{prob'}_0(1-\mathrm{prob'}_0) (2\mathrm{prob}_0-1)^2}{ \mathcal{E}^2 \lambda^{4n} }
	\end{equation*}
	We can now use our previous expression from Eq.~\eqref{lambdapower}
	for determining $\lambda^{2n}$ and  $\lambda^{4n}$,
	which finally yields our formula for the measurement costs as
	\begin{equation}
	N_{s,1} = \mathcal{O}[ \mathcal{E}^{-2 (1+f ) } ]
	\quad \quad \text{and} \quad \quad
	N_{s,2} = \mathcal{O}[ \mathcal{E}^{-2 (1+2f) } ].
	\end{equation}
	Total number of measurements required to determine
	the result is indeed $N_{s,1} + N_{s,2}$,
	and recall that $f = \frac{\ln(\lambda^{-1})}{\ln (Q^{-1})}$.

\end{proof}

\clearpage

\section{Effect of violating assumptions}

Let us now analyse how non-identical copies of $\rho$ affect the
performance of our approach.

\begin{lemma}\label{differentlemma}
	When the states are not perfectly identical via $\rho = \bigotimes_{\mu=1}^n\rho_\mu$
	with $\rho_1 \neq \rho_2 \dots \neq \rho_n$, but their dominant eigenvector is identical
	then our main result from Theorem~\ref{theo2} still holds and we still obtain exponentially
	decreasing error bounds as
	\begin{align*}
	\text{Method A:}& \quad \quad
	\frac{2\mathrm{prob}_0-1}{2\mathrm{prob'}_0-1}
	=  \langle \psi  | \sigma | \psi \rangle  +  \mathcal{O}([\lambda_\mathrm{min}^{-1}{-}1]^n),\\
	\text{Method B:}& \quad \quad
	\frac{2\mathrm{prob}_0-1}{\prod_{\mu=1}^n \lambda_\mu}
	=  \langle \psi  | \sigma | \psi \rangle  +  \mathcal{O}([\lambda_\mathrm{min}^{-1}{-}1]^n).
	\end{align*}
	For Method B we assume that the dominant eigenvalues $\lambda_1, \lambda_2, \dots \lambda_n$ are known.
	The error depends on the smallest of these dominant eigenvalues, which we denote as $\lambda_\mathrm{min}$.
	In the special case when all $\rho_\mu$ commute (i.e., same eigenvectors, but different eigenvalues)
	our error bounds $\mathcal{E}_A$ and $\mathcal{E}_B$ from Theorem~\ref{theo2} approximately
	holds via an effective sequence $Q_n^{eff}$, and we can expect an error suppression
	very similar to Theorem~\ref{theo2}.
\end{lemma}

\begin{proof}
	\textbf{Case 1}: Let us build up components of our proof by first considering the special case when all $\rho_k$
	commute with each other. In other words the states staisfy the spectral decomposition
	\begin{equation*}
	\rho_\mu =\lambda_\mu | \psi \rangle \langle \psi  |  + (1-\lambda_\mu) \sum_{k=2}^d p_{k_\mu} |\psi_{k}\rangle \langle \psi_{k} |,
	\end{equation*}
	and they all share the same eigenvectors while their eigenvalues can be different.
	It follows that the orthogonality relations in the proof of Theorem~\ref{theo1} in Eq.~\eqref{orthogonality} 
	still hold and the final result can be written explicitly as
	\begin{equation}
	\mathrm{prob}_0 =\frac{1}{2} + \frac{1}{2}  \langle \psi  | \sigma | \psi \rangle \prod_{\mu=1}^n \lambda_\mu
	+ \sum_{k=2}^d   \langle\psi_k |\sigma | \psi_k \rangle \prod_{\mu=1}^n p_{k_\mu} (1-\lambda_\mu).
	\end{equation}
	We can upper bound the error term in the above expression as
	\begin{equation}
	 \lvert \sum_{k=2}^d  \langle\psi_k |\sigma | \psi_k \rangle \prod_{\mu=1}^n p_{k_\mu} (1-\lambda_\mu) \rvert \leq
	  (1-\lambda_\mathrm{min})^n \sum_{k=2}^d  (p_{k,\mathrm{max}})^n,
	\end{equation}
	where we have denoted the largest component as $ \lambda_\mathrm{min} = \min_\mu \lambda_\mu $ and
	$ p_{k,\mathrm{max}} = \max_\mu p_{k_\mu}$ and we have also used that
	 $ \lvert \langle\psi_k |\sigma | \psi_k \rangle \rvert \leq 1$.
	
	We can also upper bound the product $\prod_{\mu=1}^n \lambda_\mu \geq \lambda_\mathrm{min}^n$
	and  derive the error of our Method A in Theorem~\ref{theo2}	which results in 
	\begin{equation*}
	\frac{2\mathrm{prob}_0-1}{2\mathrm{prob'}_0-1}
	=  \langle \psi  | \sigma | \psi \rangle  + \mathcal{E}_A,
	\quad \quad \text{with}	\quad \quad 
	|\mathcal{E}_A| \leq  \frac{2 Q_n^{eff}   }{1+  Q_n^{eff}   }
	\end{equation*}
	which we write in terms of an effective sequence
	$Q_n^{eff} = (\lambda_\mathrm{min}^{-1}{-}1)^n \lVert \underline{p}_\mathrm{max} \rVert_n^n$.
	
	We can similarly derive the errors of our Method B in Theorem~\ref{theo2}
	in case when the eigenvalues $\lambda_1, \lambda_2, \dots \lambda_n$ are known.
	This results in 
	\begin{equation*}
	\frac{2\mathrm{prob}_0-1}{ \prod_{\mu=1}^n \lambda_\mu}
	=  \langle \psi  | \sigma | \psi \rangle  + \mathcal{E}_B,
	\quad \quad \text{with}	\quad \quad 
	|\mathcal{E}_B| \leq   Q_n^{eff}   ,
	\end{equation*}
	where we have again used our effective sequence
	$Q_n^{eff} = (\lambda_\mathrm{min}^{-1}{-}1)^n \lVert \underline{p}_\mathrm{max} \rVert_n^n$.
	We note that here $\underline{p}_\mathrm{max}$ is no longer a proper probability vector since
	$\sum_{k=2}^d  p_{k,\mathrm{max}} \geq 1$ and therefore we cannot guarantee
	in general that $Q_{eff}<1$. Nevertheless, one expect a very similar exponential decay
	of the error as in Theorem~\ref{theo2} and in Lemma~\ref{renyilemma} for high-entropy
	probability distributions and for $n>1$.

	\vspace{1cm}
	\textbf{Case 2}: We now consider the most general case when
	$\rho_\mu$ are arbitrary except that their dominant eigenvector is exactly $ | \psi \rangle$.
	The states therefore admit the following spectral decompositon
	\begin{equation*}
	\rho_\mu =\lambda_\mu | \psi \rangle \langle \psi  |  + (1-\lambda_\mu) \sum_{k=2}^d p_{k_\mu} |\psi_{k_\mu}\rangle \langle \psi_{k_\mu} |,
	\end{equation*}
		It follows from the above definition  that the dominant eigenvector $| \psi \rangle $ is
	orthogonal to every error contribution in every eigenstate as $\langle \psi   |\psi_{k_\mu} \rangle = 0$ for
	every $k=\{ 2, \dots 2^N\}$ and for every $\mu= \{1, \dots n\}$.
	Modifying accordingly the orthogonality relation in the proof of Theorem~\ref{theo1} in Eq.~\eqref{orthogonality} 
	allows us to compute the leading term as expected, but every other non-zero term is multiplied with the
	prefactor $\prod_{\mu=1}^n (1-\lambda_\mu)$ which leads to the following error term
	\begin{equation}
	\mathrm{prob}_0 =\frac{1}{2} + \frac{1}{2}  \langle \psi  | \sigma | \psi \rangle \prod_{\mu=1}^n \lambda_\mu + \mathcal{O}[\prod_{\mu=1}^n (1-\lambda_\mu)].
	\end{equation}
	
	As shown  previously, this allows us to compute the error of our Method A 
	and Method B in Theorem~\ref{theo2} as
	\begin{equation*}
	\frac{2\mathrm{prob}_0-1}{2\mathrm{prob'}_0-1}
	=  \langle \psi  | \sigma | \psi \rangle  +  \mathcal{O}([\lambda_\mathrm{min}^{-1}{-}1]^n),
	\quad \quad \text{and}	\quad \quad 
	\frac{2\mathrm{prob}_0-1}{\prod_{\mu=1}^n \lambda_\mu}
	=  \langle \psi  | \sigma | \psi \rangle  +  \mathcal{O}([\lambda_\mathrm{min}^{-1}{-}1]^n),
	\end{equation*}
	in general for $\lambda_\mathrm{min} > 1/2$.
\end{proof}

\subsection{Coherent mismatch in incoherent error channels \label{coherrorsec}}

As we discussed in the main text our approach cannot address
coherent errors, i.e., when the dominant eigenvector of 
the density matrix is $ \sqrt{1-c} | \psi_{id} \rangle + \sqrt{c} | \psi_{err} \rangle$,
where $\psi_{id}$ is the ideal computational state and $\psi_{err}$ is some error.
This is expected to happen when systematic errors, such as miscalibrated
rotation angles, are present but it is straightforward to show that
even a completely incoherent error channel (random unitary events) can
introduce a slight mismatch in the eigenvectors.

We show this by considering a quite general noise channel as
\begin{equation} \label{error_channel}
 \rho' = (1-\epsilon) \rho + \epsilon \rho_{err},
\end{equation}
in which no errors happen with a probability $(1-\epsilon)$
and some error happens with a probability $\epsilon$.
In complete generality, the eigenvectors of $\rho$ can
be different than the eigenvectors of $\rho_{err}$ unless the
commutator vanishes $[\rho, \rho_{err}] =0$. A typical example
for a vanishing commutator is the single-qubit depolarising channel in single-qubit
systems, in which case $\rho' = (1-\epsilon) \rho + \epsilon \mathrm{Id}$
and indeed $[\rho, \mathrm{Id}] =0$.
However, for more than 1 qubits (or non-separable states) the
above expression does not hold and even single qubit depolarising can
introduce  a coherent mismatch such that the dominant eigenvector of $\rho'$
is $ \sqrt{1-c} | \psi \rangle + \sqrt{c} | \psi_{err} \rangle$.

The coherent mismatch due to incoherent errors is expected to be very small
in practically relevant scenarios since the high entropy of the error probabilities
from Eq.~\eqref{rhodef} ensures us that $\lVert [\rho, \rho_{err}] \rVert \ll 1$. 
For example, in our numerical simulations in Fig.~\ref{sampling} the
infidelity of the dominant eigenvector with respect to the pure
state obtained from a noise-free computation was below $10^{-4}$.

In general, for a high entropy error distribution in Eq.~\eqref{rhodef} we obtain
the spectral decomposition with $\lambda_k \ll 1$ for $k\geq2$
\begin{equation*}
	\rho = \sum_{k=1}  \lambda_k |\psi_k \rangle \langle \psi_k |.
\end{equation*}
One can compute the first order (in $\epsilon$) corrections
to the eigenvectors of $\rho'$ via the usual perturbative series:
the dominant eigenvector of $\rho'$ is approximately (up to normalisation)
\begin{equation*}
	|\psi_1' \rangle  \approx |\psi_1 \rangle + \sum_{k=2}\frac{\langle \psi_k | \epsilon \rho_{err} | \psi_1 \rangle}{ \lambda_1 - \lambda_k} | \psi_k \rangle
	=
	|\psi_1 \rangle  
	+ \frac{\epsilon}{\lambda_1} \sum_{k=2}  \langle \psi_k | \rho_{err} | \psi_1 \rangle  | \psi_k \rangle
	= |\psi_1 \rangle  + \mathcal{O}(\epsilon) ,
\end{equation*}
where we have used that $ \lambda_1 - \lambda_k \approx \lambda_1$.
Indeed the result is constant bounded due to the norm of the fist order correction
$\sum_{k=2}  |\langle \psi_k | \rho_{err} | \psi_1 \rangle  |^2 = |\mathrm{Col}_1[\rho_{err}]|^2 \leq 1$, hence
the scaling of the correction $\mathcal{O}(\epsilon)$. Here $\mathrm{Col}_1[\rho_{err}]$ is the column  vector of $\rho_{err}$
whose norm is bounded by the largest eigenvalue.

	Let us now focus on the repeated application of the noise channel from Eq.~\eqref{error_channel} which can be used to model
	a quantum circuit that applies a series of noisy quantum gates with error probability $\epsilon$.
	The incoherent decay of the dominant eigenvalue is expected to decay exponentially with $\nu$ as $(1-\epsilon)^\nu$.
	For large systems one needs to implement a large number $\nu$ of gates 
	and therefore one requires to have a sufficiently low per-gate error $\epsilon$
	in order to keep the dominant eigenvalue above a threshold $\lambda_{min}$ -- and thus
	keep the sampling costs in Result~\ref{result2} practical.
	Since the strength of the coherent mismatch is proportional to the per-gate error rate $\epsilon$,
	it is expected to decrease when we decrease $\epsilon$.
	Interestingly, our numerical simulations of the single- and two-qubit depolarising channel suggest
	that in the investigated region (see Fig.~\ref{coherror_fig}) the coherent mismatch only grows linearly
	when we increase the number of gates $\nu$.	
	This suggest that when we increase the number of gates the incoherent (exponential) decay of the
	dominant eigenvalue is significantly more damaging than the (linearly) increasing coherent mismatch.
	
	We illustrate this the following way. Let us define the following quantities. We define the fidelity between the dominant eigenvector
	$\psi_1^{(\nu)}$ (after the application of $\nu$ noisy gates)
	and the ideal state as $\eta_1:= |\langle  \psi_{id} |\psi_1^{(\nu)} \rangle |^2 = 1-c$.
	Furthermore we define the fidelity between the dominant eigenvector and the density matrix $\rho$ as
	$\eta_2 := \langle  \psi_1^{(\nu)} | \rho |  \psi_1^{(\nu)} \rangle \approx  \lambda$.
	Here $\eta_1$ decays due to the coherent mismatch while $\eta_2  \approx (1-\epsilon)^\nu$ decays purely
	due to the incoherent effect of the noise channel. Fig.~\ref{coherror_mitig} (a) shows 
	how the ratio $\eta_2/\eta_1$ decreases when we increase the number of gates. Interestingly the scale
	at which this ratio decays appears to be exponential in the investigated region.

	These results suggest that the coherent mismatch in the dominant eigenvector can be expected
	to be sufficiently small for large, complex quantum circuits. Refer to ref.~\cite{koczor2021dominant}
	for a detailed analysis.

\begin{figure*}[tb]
	\begin{centering}
		\includegraphics[width=0.95\textwidth]{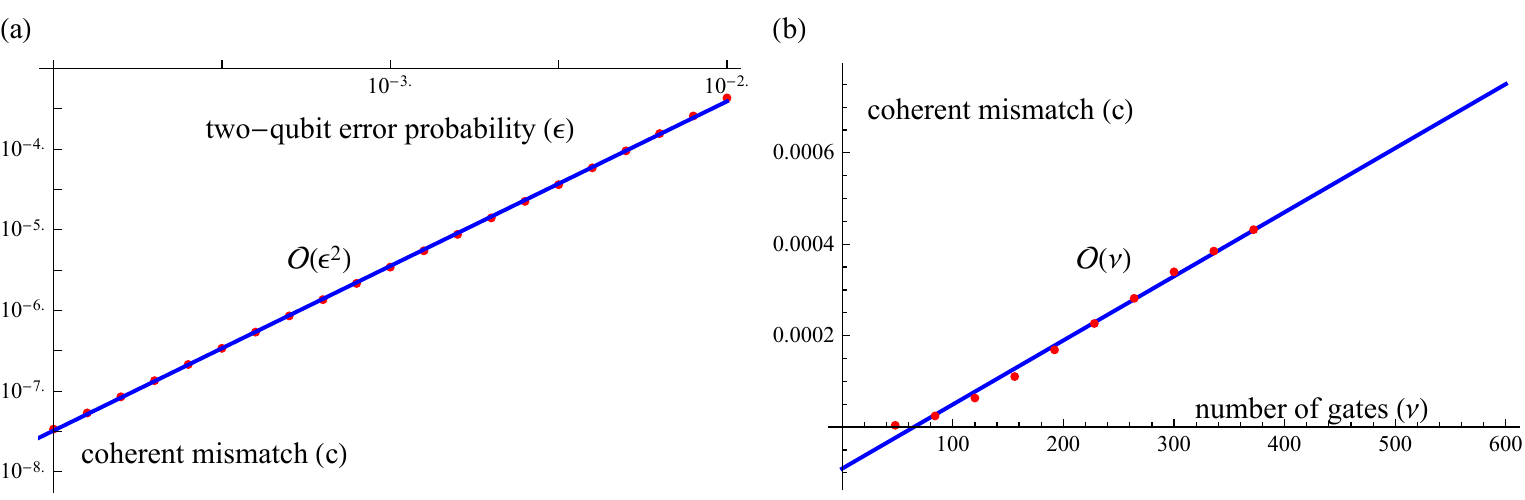}
		\caption{
			Coherent mismatch $c$ in the dominant eigenvector 
			$ \sqrt{1-c} | \psi \rangle + \sqrt{c} | \psi_{err} \rangle$
			of the density matrix $\rho$ 
			in case of a purely incoherent noise model (single and two-qubit depolarising).
			We simulated the same $12$-qubit system with
			$372$ noisy gates from  Fig.~\ref{sampling}
			in which 
			single qubit gates undergo single-qubit depolarisation with
			probability $0.1\epsilon$ while two-qubit gates undergo two-qubit depolarisation
			with probability $\epsilon$.
			(a) The coherent mismatch is small and scales with the per-gate error $\epsilon$
			as $c = \mathcal{O}(\epsilon^2)$ as explained in the text.			
			(b) The coherent mismatch is small and it grows with the number of gates $\nu$
			at a fixed gate error $\epsilon = 10^{-3}$ as $c = \mathcal{O}(\nu)$.
			\label{coherror_fig}
		}
	\end{centering}
\end{figure*}

\subsubsection{Mitigating the coherent mismatch}

As discussed in the main text, well-established techniques can be used to mitigate the effect of coherent errors.
We now focus on the above introduced coherent mismatch in the eigenvector due to incoherent error channels and
demonstrate the effectiveness of an extrapolation approach in Fig.~\ref{coherror_mitig} (b). Similarly to Fig.~\ref{extrapolation} in the main text,
we use extrapolation techniques, but here we vary the gate error rate in the state preparation stage
(and not in the derangement process). We set gate errors such that two-qubit gates undergo
a depolarising noise with probability $\epsilon=10^{-3}$ and assume that the experimentalist
can increase this error in $k=2,3,4\dots$ steps up to $\epsilon=10^{-2}$.
As expected from the above arguments based on a perturbative expansion of the
dominant eigenvector, the measured expectation value should depend on the error levels
as a polynomial that has rapidly decaying expansion coefficients due to the fact that the per-gate error level
is low as $\epsilon \ll 1$. We have determined extrapolation
errors using various fitting techniques as shown in Fig.~\ref{coherror_mitig} (b). We define the extrapolation
error as the difference between the  ideal, error free expectation value $\langle \psi_{id} | \sigma | \psi_{id} \rangle $
and the estimated expected value $\frac{\tr[\rho^n \sigma]} {\tr[\rho^n]}$ from Method A of Theorem~\ref{theo2}.
Here $ | \psi_{id} \rangle$ is the state that one would obtain from a perfect, noise-free evaluation
of the circuit and in our simulation we consider the same 12-qubit circuit as in Fig.~\ref{sampling} (right) in the
main text (refer to Appendix~\ref{numerics}) with $n=3$ copies.

Indeed, Fig.~\ref{coherror_mitig} (b) confirms that the effect of the coherent mismatch can be straightforwardly
mitigated by fitting low-order polynomials to the experimental data. The red horizontal line represents the error bound
from Result~\ref{result1} and one can demonstrably suppress the effect of the coherent mismatch below this error bound. As
expected, when we increase the degree of the fitting polynomial,
 the error saturates as it reaches the level from Result~\ref{result1} which we defined for the case when the
coherent mismatch is neglected -- and the errors could only be further suppressed by increasing the number $n$ of copies.

\begin{figure*}[tb]
	\begin{centering}
			\includegraphics[width=0.45\textwidth]{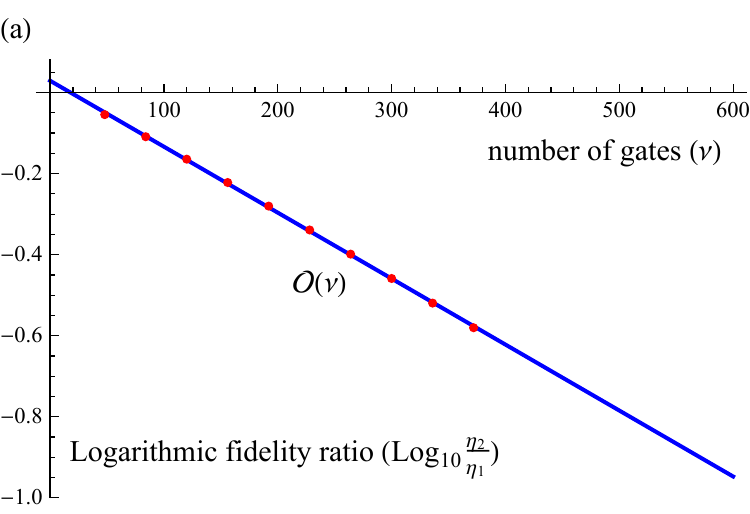}
		\includegraphics[width=0.45\textwidth]{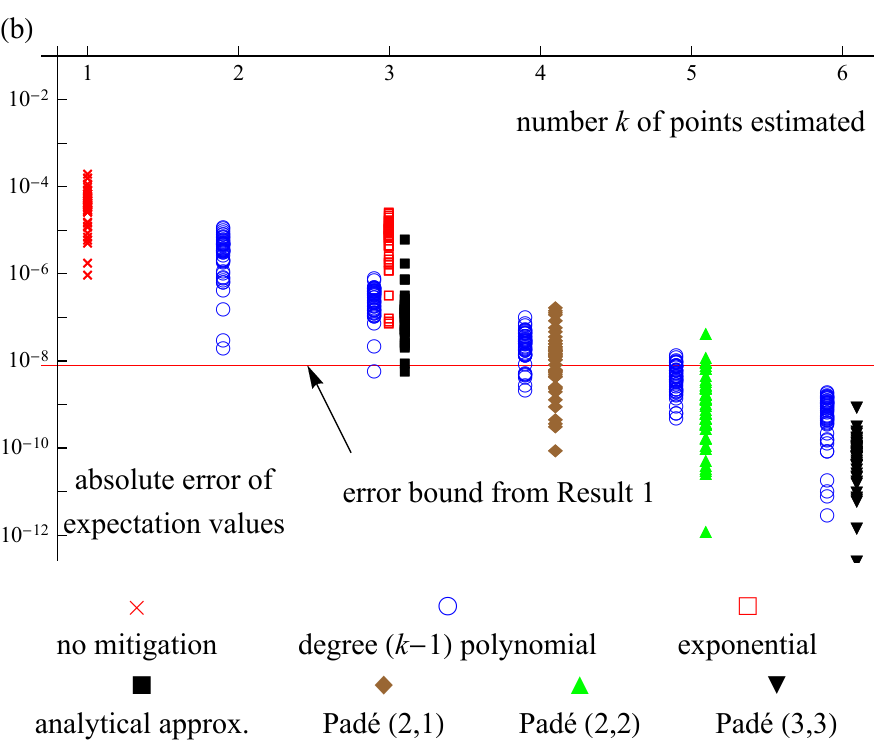}\\
		\caption{
			(a)  The fidelity $\eta_1$ decreases purely due to the coherent mismatch while the fidelity
			$\eta_2$ decays purely due to the incoherent effect of the noise channel as discussed in Appendix~\ref{coherrorsec}.
			The ratio $\eta_2/ \eta_1$ appears to decay exponentially in the investigated region
			(linear in the logarithmic plot) when we increase
			the number of gates for the fixed two-qubit gate error $\epsilon = 10^{-2}$ and this ensures
			us that the coherent mismatch in the dominant eigenvector becomes negligible for large systems
			(large $\nu$). We simulated
			the same circuit as in Fig.~\ref{coherror_fig} (b).
			(b) Mitigating the error caused by a coherent mismatch in the dominant eigenvector
			of the density matrix as discussed in Appendix~\ref{coherrorsec}. Gate error levels in the
			state preparation stage were varied in $k=2,3,4\dots$ steps and the obtained
			expectation values were extrapolated to the zero error limit.  The red horizontal line
			represents the error bound from Result~\ref{result1} and one can straightforwardly
			suppress the effect of the coherent mismatch below this error bound
			by fitting low-order polynomials. The error saturates when increasing the degree
			of the fitting polynomial and can only be further reduced by increasing the number $n$ of
			copies.
			\label{coherror_mitig}
		}
	\end{centering}
\end{figure*}

\clearpage

\section{Noise resilience of derangements and error extrapolation \label{noise_resilience_sec}}

\begin{example}\label{error_example}
	\normalfont
	We now show examples why the derangement operator is highly resilient to errors.
	We proceed by recapitulating that quantum channels can be represented by a set of
	non-unique Kraus maps and, in particular, we consider the decomposition into the following sum
	of unitary transformations as
	\begin{equation*}
	\rho' = (1-\epsilon) U \rho U^\dagger +\epsilon \sum_{m} c_m U_m \rho U^\dagger_m,
	\end{equation*}
	where $U$ is the ideal unitary transformation, $\sum_{k} c_m = 1$ and $ 0 \leq \epsilon, c_m \leq 1$,
	while  the erroneous Kraus operators
	are unitary via $U_m  U^\dagger_m = \mathrm{Id}$.
	The deviation from the ideal transformation can be interpreted as unitary transformations
	that randomly affect the eigenvectors $|\psi\rangle$ of the
	quantum state as $U_m |\psi\rangle$ with probability
	$\epsilon c_m$.
	
	Let us now analyse how such errors affect our procedure when they occur 
	during the derangement operator, i.e., we set the ideal transformation $U$ to be our derangement circuit from Fig.~\ref{schem}.
	First, we show that the orthogonality
	relations in the proof of Theorem~\ref{theo1} are resilient to such noise events.
	In particular, recall that the derangement operator symmetrises the input state
	as, e.g.,
	$$  \perm | \psi_1 , \psi_2 , \dots \psi_n \rangle = |\psi_n , \psi_1  , \dots \psi_{n-1}\rangle,$$
	which would ideally ensure that only permutation-symmetric combinations contribute to the
	output via the orthogonality relation from Eq.~\eqref{orthogonality} as
	\begin{equation*} 
	\langle \psi_{k_1}, \psi_{k_2},  \dots \psi_{k_n}	| \sigma \psi_{k_n}, \psi_{k_1},  \dots \psi_{k_{n-1}} \rangle
	=\langle\psi_{k_1} |\sigma \psi_{k_n} \rangle \langle  \psi_{k_2} | \psi_{k_1} \rangle \cdots \langle \psi_{k_n} | \psi_{k_{n-1}} \rangle.
	\end{equation*}
	One can show that even if errors occur during the derangement procedure
	the orthogonality relations are still preserved as
	\begin{equation*} 
	\langle U_{1} \psi_{k_1}, U_{2} \psi_{k_2},  \dots U_{n} \psi_{k_n}	| \sigma U_{1} \psi_{k_n}, U_2 \psi_{k_1},  \dots U_n \psi_{k_{n-1}} \rangle
	=\langle U_{1} \psi_{k_1} |\sigma U_{1} \psi_{k_n} \rangle \langle  \psi_{k_2} | \psi_{k_1} \rangle \cdots \langle \psi_{k_n} | \psi_{k_{n-1}} \rangle.
	\end{equation*}
	It follows that the non-symmetric combinations of input states do
	not contribute to the output even when the derangement operator is
	affected by random errors. Note that even though the errors do not directly contribute to
	the final output (as shown above), the probability that the circuit outputs an error-free result
	is decreased via the $1-\epsilon$ factor. This is, however, a trivial effect that only attenuates the output
	probabilities linearly and can be completely corrected by a linear extrapolation (i.e., estimating
	the output probabilities at different $\epsilon$ values and then extrapolating to $\epsilon = 0$).
	
	Second, let us show that for symmetric input states $| \psi, \psi,  \dots \psi \rangle$
	all random errors during the derangement procedure cancel that do
	not affect the ancilla qubit nor the register to which the observable $\sigma$
	is applied. In fact, we just modify the above equation by not allowing
	errors on register $1$ as
	\begin{equation*} 
	\langle \psi, U_{2} \psi,  \dots U_{n} \psi	| \sigma  \psi, U_2 \psi,  \dots U_n \psi \rangle
	=\langle \psi |\sigma \psi \rangle \langle  \psi| \psi \rangle \cdots \langle \psi | \psi \rangle
	= \langle \psi |\sigma | \psi \rangle.
	\end{equation*}	
	The second equation shows that we obtain the correct contribution
	despite all registers except for register $1$ have undergone some
	random error $U_2, U_3$ etc. Our previous argument again holds:
	despite the fact that these error events do not directly contribute to the final output of the circuit,
	the probability of an error-free output is attenuated linearly which,
	nonetheless, can be completely corrected by a linear extrapolation.
	
	In summary, the derangement measurement is highly
	resilient to errors and completely protects the permutation
	symmetry of input states even when the derangement operator
	suffers from experimental noise. However, errors that affect the
	qubits to which the observable $\sigma$ is applied will degrade
	the final result \emph{non-trivially} via $\langle U_1 \psi |\sigma | U_1 \psi \rangle$,
	where $U_1$ is some unitary noise process that occurs with a (possibly)
	low probability. 	
	Nevertheless, we show in the main text and
	in the following theorem that these erroneous contributions can be successfully
	mitigated with, e.g.,  extrapolation techniques.

\end{example}

\begin{theorem} \label{pade_theo}
	Assume that a circuit consists of a sequence of $\nu$ noisy quantum gates, and
	each gate's error model is of the form $(1-\epsilon) \Phi_{k} + \epsilon \mathcal{E}_k$,
	where $\Phi_{k}$ is the ideal, error-free quantum channel and $\mathcal{E}_k$ is 
	an arbitrary error channel (CPTP map) that occurs with probability $\epsilon$. Most typical error models 
	are of this form, including dephasing, depolarising, inhomogeneous Pauli errors,
	arbitrary unital channels and beyond ($\mathcal{E}_k$ need not be local or two-local).
	In a circuit that consists of number $\nu$ such gates, any expectation value $E$ will depend on the error probability $\epsilon$ as a
	degree $\nu$ polynomial via
	\begin{equation*}
	E (\epsilon)=	E_{0} + \sum_{k=1}^\nu \epsilon^k \, E_{k},
	\end{equation*}
	where $E_{k}$ are real polynomial coefficients. One can therefore exactly
	determine the ideal expectation value $E (0)$ by estimating $E (\epsilon)$
	at $\nu+1$ points in $\epsilon$. The so-called Lagrange polynomial or the Newton polynomial
	provide explicit formulas for computing $E(\epsilon)$ from
	the pointwise reconstructions $E(\epsilon_k)$. Furthermore, one can approximate the dependence
	on $\epsilon$  via, e.g.,  the $(3,3)$ Pad{\'e} approximation as
	\begin{equation}
	E (\epsilon) \approx  E_0  - \tilde{\eta} \epsilon	\, \frac{(1 - \epsilon)^{ \nu} }{ 2  \epsilon-1 }
	\approx	\frac{a_1 \epsilon  + a_2 \epsilon^2 + a_3 \epsilon^3
	}{ 1 + a_4 \epsilon + a_5 \epsilon^2
	},
	\end{equation}
	that only requires the coefficients $a_1, a_2, a_3, a_4, a_5$ to be
	fitted to experimental data.
\end{theorem}
\begin{proof}
	Applying $\nu$ gates in a sequence will result in the
	product of channels
	\begin{equation}
	\prod_{k=1}^\nu [ (1-\epsilon) \Phi_{k} + \epsilon \mathcal{E}_k ]
	 =
	 (1-\epsilon)^\nu \prod_{k=1}^\nu \Phi_k + \sum_{k=1}^\nu (1-\epsilon)^{\nu-k} \epsilon^k \mathcal{G}_k,
	\end{equation}
	where $\mathcal{G}_k$ is a channel which decomposes into the sum of all terms in which $k$ errors occur and
	$\prod_{k=1}^\nu \Phi_k $ is the ideal error-free circuit. 	We can introduce the circuit with no errors as
	$\mathcal{G}_0:= \prod_{k=1}^\nu \Phi_k$ which simplifies our formula as.
	\begin{equation}
	\prod_{k=1}^\nu [ (1-\epsilon) \Phi_{k} + \epsilon \mathcal{E}_k ]
	=
	 \sum_{k=0}^\nu (1-\epsilon)^{\nu-k} \epsilon^k \mathcal{G}_k,
	\end{equation}
	
	It follows that any expectation value (with respect to some observable $\mathcal{H}$) will be of the form
	\begin{equation}
	E = \tr\{ \mathcal{H}	\prod_{k=1}^\nu [ (1-\epsilon) \Phi_{k} + \epsilon \mathcal{E}_k ] \rho  \}
	=\sum_{k=0}^\nu (1-\epsilon)^{\nu-k} \epsilon^k \, \tr \{ \mathcal{H} \mathcal{G}_k  \rho \},
	\end{equation}

	therefore any expectation value can be expressed as a degree $\nu$ polynomial
	as a function of the error probability as
	\begin{equation}
	E (\epsilon)=
	E_{0} + \sum_{k=1}^\nu \epsilon^k \, E_{k},
	\end{equation}
	where $E_0$ is the ideal, noise-free expectation value
	and $E_k$ are polynomial coefficients.
	
	Let us now write (without loss of generality) that the expectation values
	are of the form $\tr \{\mathcal{H} \mathcal{G}_k  \rho \} =  \tilde{\eta} + \eta_k$,
	where $\tilde{\eta}$ is a mean value and $\eta_k$ expresses the deviation
	from the mean value. Let us assume that $\eta_k \ll \tilde{\eta}$, which
	in the case of the derangement operator is motivated by our argument in \ref{error_example},
	that most errors do not contribute and therefore $\tilde{\eta} \approx 0$.
	In this case we can evaluate the summation analytically for the mean value 
	\begin{equation}
	E (\epsilon) = E_0 + \sum_{k=1}^\nu (1-\epsilon)^{\nu-k} \epsilon^k \, \tr \{ \mathcal{H} \mathcal{G}_k  \rho \}
	= E_0 +  \tilde{\eta}	\sum_{k=0}^\nu (1-\epsilon)^{\nu-k} \epsilon^k + \mathcal{O}(\eta_k)
	=  E_0 + \tilde{\eta} \epsilon	\, \frac{ \epsilon^{\nu} -(1 - \epsilon)^{ \nu} }{ 2  \epsilon-1 } + \mathcal{O}(\eta_k).
	\end{equation}
	We can obtain a Pad{\'e} expansion of the above result at $\epsilon \approx 0$
	by neglecting the term $\epsilon^{\nu}$. For example the $(3,3)$ Pad{\'e} approximation follows as
	\begin{equation}
	E (\epsilon) \approx  E_0  - \tilde{\eta} \epsilon	\, \frac{(1 - \epsilon)^{ \nu} }{ 2  \epsilon-1 }
	\approx E_0  - \tilde{\eta}	\, 
	\frac{ \epsilon  + a(n) \epsilon^2 + b(n) \epsilon^3
	}{ 1 + c(n) \epsilon + d(n) \epsilon^2
	},
	\end{equation}
	where $a(n), b(n), c(n), d(n)$ are the Pad{\'e} expansion coefficients that depend on the number $n$ of gates $\nu$,
	for example
	\begin{eqnarray}
	a(n) = \frac{2 (62 + 11 n - 8 n^2 + n^3)}{5 (26 - 9 n + n^2)}.
	\end{eqnarray}
	Indeed this expansion is only valid when $\eta_k \approx 0$. Nevertheless, we propose to
	approximate the polynomial
	\begin{equation}
	E (\epsilon) \approx
	\frac{a_1 \epsilon  + a_2 \epsilon^2 + a_3 \epsilon^3
	}{ 1 + a_4 \epsilon + a_5 \epsilon^2
	},
	\end{equation}
	by fitting the coefficients $a_1, a_2, a_3, a_4, a_5$ to
	experimental data.
\end{proof}

\clearpage

\section{Hardware-native implementation of derangement circuits}

\begin{table*}[tb]
	\centering 
	\caption{\label{recompile}
			Number $\nu_e$ of entangling gates with 2(3)-qubit gates in bold (Roman) and number $\nu_s$ of single-qubit
			gates when recompiling elementary controlled-SWAP operations:
			fully equivalent (type A), local $SU(4)$ equivalent recompilation (type B)
			and recompilation with including the observable $\sigma$ (type C). 
			 $R_{\alpha}$ denote single-qubit rotation gates with $\alpha \in \{x,y,z\}$,
			while $\mathrm{C}[R_\alpha]$ denote controlled-rotations.
			$R_{xx}$ denotes the xx gate and $\text{pSWAP}$ is a parametrised SWAP gate. Three qubit
			gates in the last two rows are the xxx gate $R_{xxx}$ and the controlled-controlled-phase
			gate $\mathrm{CC}[\mathrm{P}]$.
			See Appendix~\ref{appendix_recompile} for more details.
	}
	\begin{tabular}{@{\hspace{2mm}}l@{\hspace{7mm}}c@{\hspace{1.5mm}}c@{\hspace{7mm}}c@{\hspace{1.5mm}}c@{\hspace{7mm}}c@{\hspace{1.5mm}}c@{\hspace{2mm}} }
		\\[-3mm]
		\hline\hline
		\\[-4.5mm]
		&	type   &   A  &   type & B	& type & C		\\[-1mm] 	
		native gateset  &   $\nu_e$ &   $\nu_s$		&  $\nu_e$ & $\nu_s$ &  $\nu_e$ & $\nu_s$
		\\[0.5mm] 
		\hline 
		\\[-3mm]	
		$\mathrm{C}[R_x]$, $R_{y,z}$ & \textbf{6} & 6 & \textbf{5} & 2 & \textbf{4} & 2 \\[1mm]
		$\mathrm{C}[R_z]$, $R_{x,y,z}$ & \textbf{6} & 15 & \textbf{5} & 4 & \textbf{4} & 4 \\[1mm]
		$R_{xx}$, $R_{y,z}$ & \textbf{6} & 11 & \textbf{5} & 6 & \textbf{4} & 6 \\[1mm]
		$\text{pSWAP}$, $R_{x,y,z}$ & \textbf{6} & 11 & \textbf{5} & 4 & \textbf{4} & 3 \\[1.5mm]
		\hline\\[-3mm]
		$R_{xxx}$, $R_{xx}$, $R_{y,z}$ & 3{+}\textbf{3} & 10 & 3 & 6 & 2{+}\textbf{1} & 6 \\[1mm]
		$\mathrm{CC}[\mathrm{P}]$, $R_{x,y}$, $\mathrm{C}[R_z]$ & 1{+}\textbf{2} & 6 &  1{+}\textbf{1} & 3 &  1{+}\textbf{2} & 3
		\\[2mm] \hline \hline
	\end{tabular} 
\end{table*}

\subsection{Recompiling controlled-SWAP gates\label{appendix_recompile}}

Recall that derangement circuits permute registers via Definition~\ref{derangement_def}.
Permuting two registers is performed by the SWAP operator, which decomposes into
a product of $N$ elementary, two-qubit SWAP gates as $\text{SWAP}_{N,N'} \cdots \text{SWAP}_{2,2'} \text{SWAP}_{1,1'}$,
where $N$ is the number of qubits in a  register.
Our aim is now to optimally recompile elementary, controlled-SWAP gates assuming various different hardware-native gatesets.

The controlled-SWAP, also called Fredkin, gate has been much investigated in the
literature, but mostly in the context of fault-tolerant quantum computing. For example,
ref.~\cite{nguyen2013space} provided a circuit that optimally implements the controlled-SWAP gate using
$8$ applications of CNOT gates and $9$ applications of $T$ gates.
Early works have suggested that if one has the ability to natively implement
\emph{any} two-qubit gate, then one can in principle implement the controlled-SWAP
gate with only $5$ applications of arbitrary two-qubit gates \cite{PhysRevLett.75.748,smolin1996five}.
These works have provided circuit representations using $7$ applications of controlled-X
rotation gates. We now use general techniques of ref.~\cite{khatri2019quantumassisted} to
recompile the controlled-SWAP gate and find more compact representations under the assumption that
only a limited set of hardware-native gates can be executed by the hardware. We also find 
analytical guarantees that the recompiler has found the most compact representation
possible. Results as the number $\nu_e$ of entangling gates and number $\nu_s$ of single-qubit
gates are summarised in Table~\ref{recompile}. While the corresponding detailed circuits
can be found online \cite{Note\thefootcombined,git_derangement_circuits}, we show the resulting circuit diagrams
in Fig.~\ref{circuits1}, Fig.~\ref{circuits2} and Fig.~\ref{circuits3}.

\textbf{Equivalence classes:} Before discussing details of the recompilation, let us recapitulate basic definitions. A unitary $U$
is fully recompiled into $V$ if their action on every quantum state in the Hilbert space $\mathcal{H}$
is identical, i.e., there exists a global phase factor freedom $\exists \phi \in \mathbb{R}$ such that
\begin{equation*}
	U | \psi \rangle = e^{-i\phi} V |\psi \rangle, \quad \quad \forall \psi \in \mathcal{H}.
\end{equation*}
We also consider the case of local $SU(4)$ equivalent recompilations $V'$ which, in contrast, 
result in equivalence only up to a local $SU(4)$ transformation as
\begin{equation*}
	\exists W \in SU(4):\quad 
	U | \psi \rangle = e^{-i\phi} W V' |\psi \rangle, \quad \quad \forall \psi \in \mathcal{H}.
\end{equation*}
We apply this definition to the case of elementary controlled-SWAP gates where the $SU(4)$ transformation
acts locally on the two swapped qubits.
The reason why these circuits are important is the following. We notice that after the derangement circuit
$D_n$ and the observable $\sigma$ in Fig.~\ref{schem}
we can apply any local unitary transformation $W$ to the quantum registers without changing the outcome of the measurement
on the ancillary qubit. This generally allows us to recompile those controlled-SWAP gates into more compact circuits
that are not followed by any further operations.
As such, when considering $n=2$ copies, the entire derangement circuit can be recompiled into these more
compact circuits.
Let us now introduce the 3 types of recompilations used in this work.

\textbf{Type A:}
We consider the fully equivalent recompilation which may be necessary when $n>2$ and when implementing controlled-SWAP
operations that are followed by other controlled-SWAP operators acting on the same registers. The first column of
Table~\ref{recompile} shows that we generally need $6$ two-qubit operations to implement an elementary controlled-SWAP gate
and corresponding compact circuits are illustrated in Fig.~\ref{circuits1} and Fig.~\ref{circuits2}.
We also need to consider fully equivalent recompilation if the observable is measured by post-selecting on the ancilla and sampling
the output of the registers -- and not by implementing the controlled-observable as in Fig.~\ref{schem}. The former scheme would
allow us to estimate multiple observables simultaneously.

\begin{figure*}[tb]
	\begin{centering}
		\includegraphics[width=0.9\textwidth]{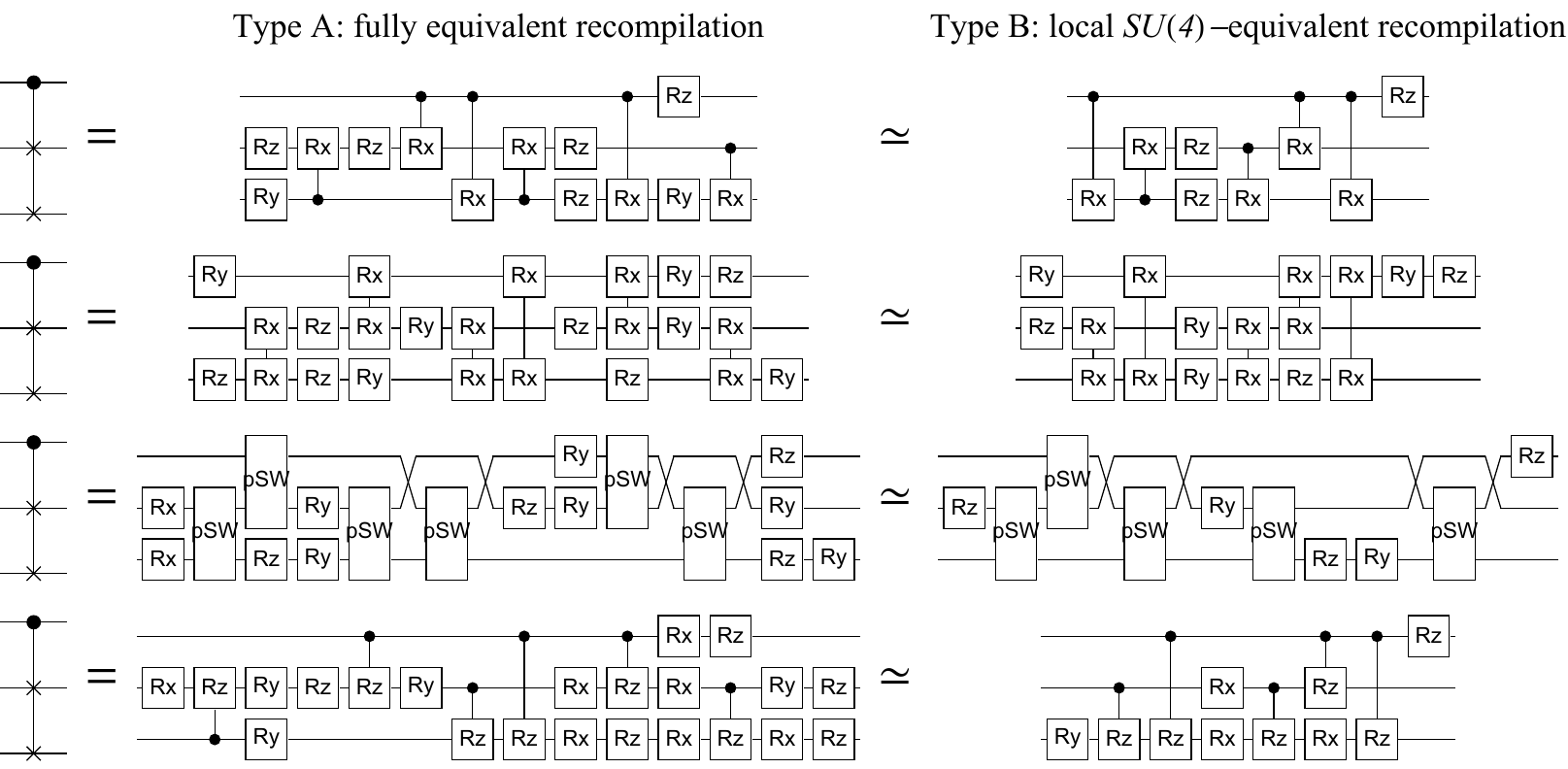}
		\caption{
			Recompiling the controlled-SWAP gate into hardware-native gates.
				In all circuits the the last $R_z$ rotation gate on the control qubit can be removed as it commutes with the
				controlled-SWAP gate and can be merged with the basis transformation of the ancilla qubit (Hadamard gate)
				immediately prior to measurement in Fig.~\ref{schem}.
			\label{circuits1}
		}
	\end{centering}
\end{figure*}

\textbf{Type B:}
If a controlled-SWAP gate is not followed by any other gate we only need to recompile up to a local $SU(4)$ freedom.
For example in case of $n=3$ copies and observable $\sigma=\mathrm{Id}$, the derangement circuit consists of two swaps
of pairs of registers as
$\mathrm{C}[\text{SWAP}_{2,3} \text{SWAP}_{1,2}]$. We need to consider Type A recompilation for $\mathrm{C}[\text{SWAP}_{1,2}]$,
and we can consider type B recompilation for $\mathrm{C}[\text{SWAP}_{2,3}]$, since the latter is not followed by any other operation
on the main registers. The second column of Table~\ref{recompile} shows that we generally need $5$ two-qubit operations to implement such an elementary gate.
The resulting compact circuits are illustrated in Fig.~\ref{circuits1} and Fig.~\ref{circuits2}.
We note that the entire $\mathrm{C}[\text{SWAP}_{2,3} \text{SWAP}_{1,2}]$
circuit could also be recompiled up to an $SU(8)$ freedom.

\textbf{Type C:}
Observables as Pauli strings $\sigma \in \{\mathrm{Id}_2, X, Y, Z\}^{\otimes N}$ act on some or all of the qubits
non-trivially. For example, consider $n=2$ copies and the Pauli string $X_4 Y_9$, in which case we can consider Type B recompilation
for all controlled-SWAP operators except for the ones that swap qubits $4$ with $4'$ and $9$ with $9'$ in the two registers.
Similarly, we need only recompile the
product of the elementary swap and the observable, $\mathrm{C}[X_4 \, \text{SWAP}_{4,4'}]$ and $\mathrm{C}[Y_9 \, \text{SWAP}_{9,9'}]$,
up to a local $SU(4)$ freedom. In the present work we fix an $X$ basis: one can thus implement
$\mathrm{C}[\text{SWAP}_{9,9'} Y_9]$ by first transforming the basis of qubits $9$ and $9'$ using single-qubit rotations.
The third column of Table~\ref{recompile} shows that we generally need $4$ two-qubit operations to implement such an elementary gate.
The resulting compact circuits are illustrated in Fig.~\ref{circuits3}.

\textbf{Gatesets:}
Let us start by defining single qubit $R_\alpha(\theta)$ rotation gates that depend on a parameter
that can be calibrated to any fixed value in experiments $-2\pi \leq \theta \leq 2\pi$ as
\begin{equation*}
	R_\alpha(\theta) := \exp[-i \frac{\theta}{2} \sigma_\alpha ] \quad \quad
	\text{with} \quad \quad
	\sigma_\alpha \in \{X,Y,Z\},
\end{equation*}
where $X$, $Y$ and $Z$ are Pauli matrices.

\emph{$\mathrm{C}[R_x]$ gates---}Let us first consider the aforementioned case of controlled-X rotation gates. We assume that the
hardware can natively implement single-qubit $Y$ and $Z$ rotations as well two-qubit
controlled-X rotations which we define as
\begin{equation*}
	\mathrm{C}[R_x(\theta)] := |0\rangle \langle 0| \otimes \mathrm{Id}_2 + |1\rangle \langle 1| \otimes R_x(\theta).
\end{equation*}
This unitary generates the CNOT gate when $\theta = \pi$. The controlled-SWAP gate can then be implemented via
$6$ ($5$) applications of the controlled-$x$ rotation gate in case of the fully (locally) equivalent recompilation
as illustrated in the first row of Fig.~\ref{circuits1}. Refer to the last row of Fig.~\ref{circuits1} for the
similar case of controlled $R_z$ rotation gates.

\begin{figure*}[tb]
	\begin{centering}
		\includegraphics[width=0.9\textwidth]{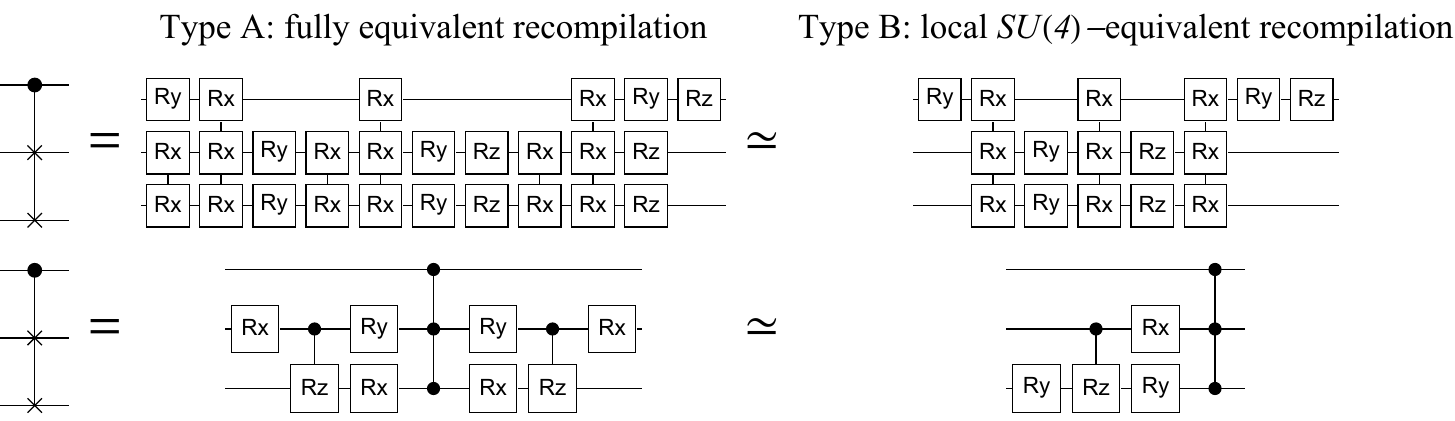}
		\caption{
			Recompiling the controlled-SWAP gate assuming that the hardware can natively
			implement 3-qubit gates, such as the $xxx$ gate (first row) and the
			controlled-controlled-phase gate (second row).
			\label{circuits2}
		}
	\end{centering}
\end{figure*}

\emph{$R_{xx}$ gates---}For the simulations in Fig.~\ref{groundstate} we assume that the hardware can natively
implement single-qubit $Y$ and $Z$ rotations as well as the two-qubit $XX$ gate of the form
\begin{equation*}
	R_{xx}(\theta) := \exp[-i \frac{\theta}{2} X \otimes X ],
\end{equation*}
which generates the M{\o}lmer-S{\o}rensen gate at $\theta=-\pi/2$. The controlled-SWAP gate can be implemented via
$6$ ($5$) applications of the XX rotation gate in case of the fully (locally) equivalent recompilation
as illustrated in the second row of Fig.~\ref{circuits1}.

\emph{$pSWAP$ gates---}We now consider an entangling gate that depends on two parameters $\theta_1$ and $\theta_2$
as
\begin{equation*}
	\text{pSWAP}[\theta_1,\theta_2] := \exp[-i \frac{\theta_1}{2}( X \otimes X + Y \otimes Y)  -i \frac{\theta_2}{2} Z \otimes Z],
\end{equation*}
which is typical to superconducting systems.
This gate is locally equivalent to the fermionic simulation gate from \cite{PhysRevLett.125.120504}
and generates at special angles many important gates, such as the SWAP gate. Although being a more general two-qubit
gate than the ones above, using the parametrised SWAP does not result in a significant improvement when recompiling
the controlled-SWAP gate: We still need
$6$ ($5$) applications of the entangling gate in case of the fully (locally) equivalent recompilation
as illustrated in the third row of Fig.~\ref{circuits1}.

\emph{$xxx$ gates---}Let us now turn to the question: can we obtain more compact representations
of the controlled-SWAP gate when we assume that the hardware can natively implement three-qubit gates.
Let us first consider the case when we allow both $xx$ and $xxx$ gates, where we define the latter as
\begin{equation*}
	R_{xxx}(\theta) := \exp[-i \frac{\theta}{2} X \otimes X \otimes X  ].
\end{equation*}
In this case we can analytically solve the recompilation problem by recalling that
the controlled-SWAP gate can be expressed as
\begin{equation*}
	\mathrm{C}[\text{SWAP}] = \exp[-i \frac{\pi}{8} (\mathrm{Id}_2 -Z) \otimes ( X \otimes X + Y \otimes Y + Z \otimes Z - \mathrm{Id}_2 \otimes \mathrm{Id}_2  )  ].
\end{equation*}
Since above all terms in the exponential commute, we can express this gate as the following series
of gates executed in arbitrary order: a $zxx$ gate, a $zyy$ gate, a $zzz$ gate, an $xx$ gate, a $yy$ gate,
a $zz$ gate and additionally a $z$ gate on the ancilla, which however can be removed as discussed below Fig.~\ref{circuits1}.
Indeed, our recompiler has found exactly this kind of circuit by mapping the $xxx$ gate to, e.g., the $zyy$ gate via single-qubit
rotations. Refer to the first row of Fig.~\ref{circuits2}. Since all these multi-qubit gates commute, we can order them such that
all the $xx$, the $yy$ and the $zz$ gates are at the end of the circuit. These gates then form a local $SU(4)$ unitary that
can be removed as discussed above. Indeed, our recompiler has found this solution when considering only locally
equivalent recompilations as illustrated in Fig.~\ref{circuits2}.
\emph{$\mathrm{CC}[P]$ gates---}Let us finally consider controlled-controlled phase gates, which have been
successfully implemented in experiments as native gates \cite{fedorov2012implementation}. Let us define this gate as
\begin{equation*}
	\mathrm{CC}[P] := \mathrm{diag}(1,1,1,1,1,1,1,-1).
\end{equation*}
These allow for surprisingly compact representations: the controlled-SWAP gate can be implemented using a single application
of the $\mathrm{CC}[P]$ gate plus $2$($1$) applications of controlled-Z rotation gates as illustrated in the second row
of Fig.~\ref{circuits2}

\begin{figure*}[tb]
	\begin{centering}
		\includegraphics[width=0.9\textwidth]{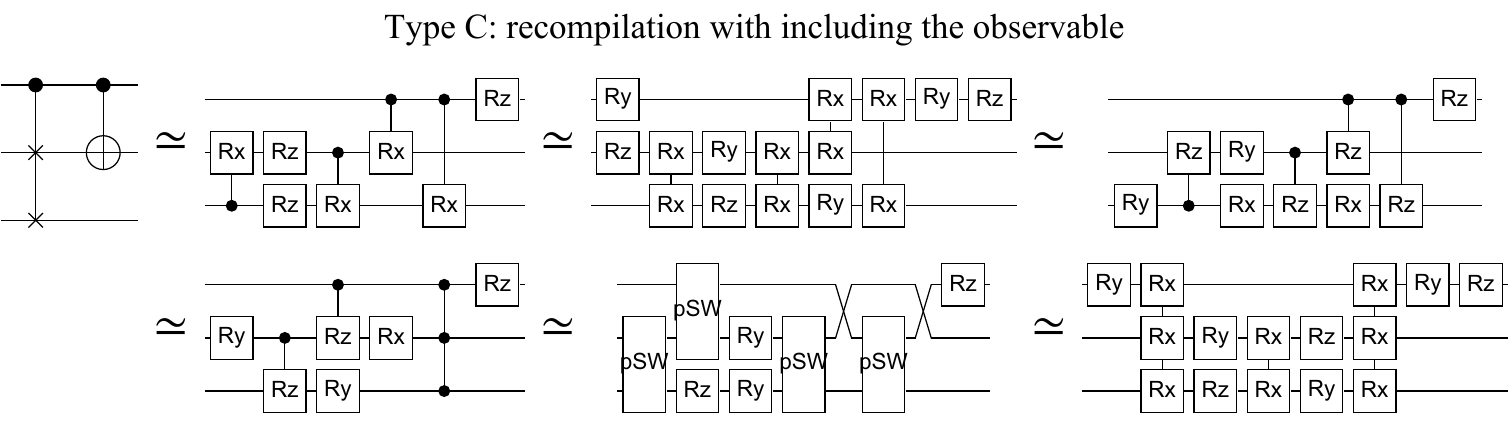}
		\caption{
			Type C: recompiling the product of the controlled-SWAP gate and the controlled observable $\sigma$ up to
			a local $SU(4)$ freedom on the swapped qubits.
			\label{circuits3}
		}
	\end{centering}
\end{figure*}

\subsection{Exploiting symmetries in derangement circuits \label{symmetries}}
As discussed in the main text, derangement circuits have a rapidly growing number of invariants
	when increasing the number of copies $n$ or the number of qubits $N$.
This includes the large number of distinct permutations from Definition~\ref{derangement_def}.
Let us illustrate how these symmetries can be exploited via the following three examples.

\textbf{Example 1:}
	In the first example we assume that the connectivity between registers is limited
	such that only nearest neighbour registers can be swapped. For example, this could be a quantum device
	with $n$ individual quantum processors arranged in a line. A derangement operation in this case is preferred
	that swaps only nearest-neighbour registers. For example, for $n=4$ we can first apply nearest neighbour SWAP operators
	between registers as
	$\text{SWAP}_{1,2}$ and $\text{SWAP}_{3,4}$ followed by $\text{SWAP}_{2,3}$. In contrast, a derangement
	which uses, e.g., $\text{SWAP}_{1,4}$ would not be supported natively by the device.
	Refer to \cite{Note\thefootcombined,git_derangement_circuits} for a demonstration of the various distinct derangement circuits.

\textbf{Example 2:}
In the second example we assume that arbitrary connectivity is available. In this case one 
can in principle implement any of the distinct derangement (permutation) patterns. Although
choosing and fixing any one of those is sufficient, it is also possible to randomly choose from these circuits
since noise may affect them differently. This may also help in suppressing asymmetries in
the potentially non-identical input density matrices thereby creating randomised, `average' density matrices
when $n$ is large.

Note that further invariants exist due to the fact that the $\sigma$ gate in Fig.~\ref{schem} can be applied
to any of the registers.

\textbf{Example 3:}	
In the third example let us consider an approach that can exploit symmetries in the derangement operator
in complete analogy with twirling techniques.
We notice that the controlled-derangement operator is a very
specific kind of operation: when the ancilla state is $|0\rangle$ then the registers are left invariant
and when the ancilla state is $|1\rangle$ then the registers are permuted.
We start by applying an operation $U$ before the derangement
operator, i.e., Pauli strings $P_k \in \{ \mathrm{Id}_2, X, Y, Z\}^{\otimes N}$ are applied to registers $k$
with  $U:= \otimes_{k=1}^n P_k$.
We can undo this operation after the derangement operator by first applying the anti-controlled Pauli string $U$.
This is then followed by the controlled Pauli string $U'$, where
in $U'$ we need to relabel the indexes $k$ according to the permutation $s(k)$ that the derangement
operator implements, i.e., $U':= \otimes_{k=1}^n P_{s(k)}$.

Applying the above (controlled) Pauli strings before (after) the controlled-derangement operator
does not affect the expectation-value measurement in an ideal scenario.
However, the Pauli strings in $U$ applied to the registers do reflect the errors that happen during the swap process,
thus randomly applying Pauli strings and averaging the measurement results can reduce and homogenise
the impact of errors that happen in the derangement circuit. Note that this is analogous to twirling techniques.

Let us illustrate this on the particular case of $n=2$ copies.
We randomly select two Pauli string $P_1,P_2 \in \{ \mathrm{Id}_2, X, Y, Z\}^{\otimes N}$,
where we apply $P_1$ to the first input state $\rho$ and apply $P_2$ to the second register.
Note that these are just single-qubit $X$, $Y$ and $Z$ operations (or the identity) applied
to individual qubits in the registers.
We then perform the controlled-derangement operator $D_2$ which swaps the two registers.
We can now undo the effect of Pauli strings by applying the anti-controlled Pauli string $P_1$ and
$P_2$ to registers $1$ and $2$, respectively. Since the derangement operator
swaps the two registers we apply the controlled $P_1$ operator to the second register and
the controlled $P_2$ operator to the first register.

\begin{figure*}[tb]
	\begin{centering}
		\includegraphics[width=0.45\textwidth]{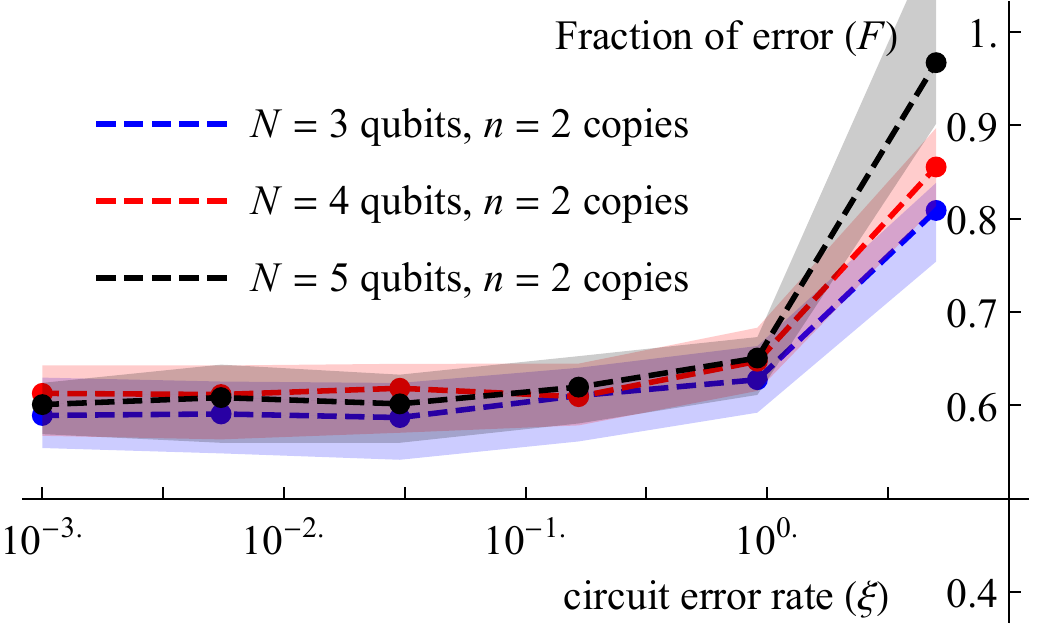}\hspace{10mm}
		\includegraphics[width=0.45\textwidth]{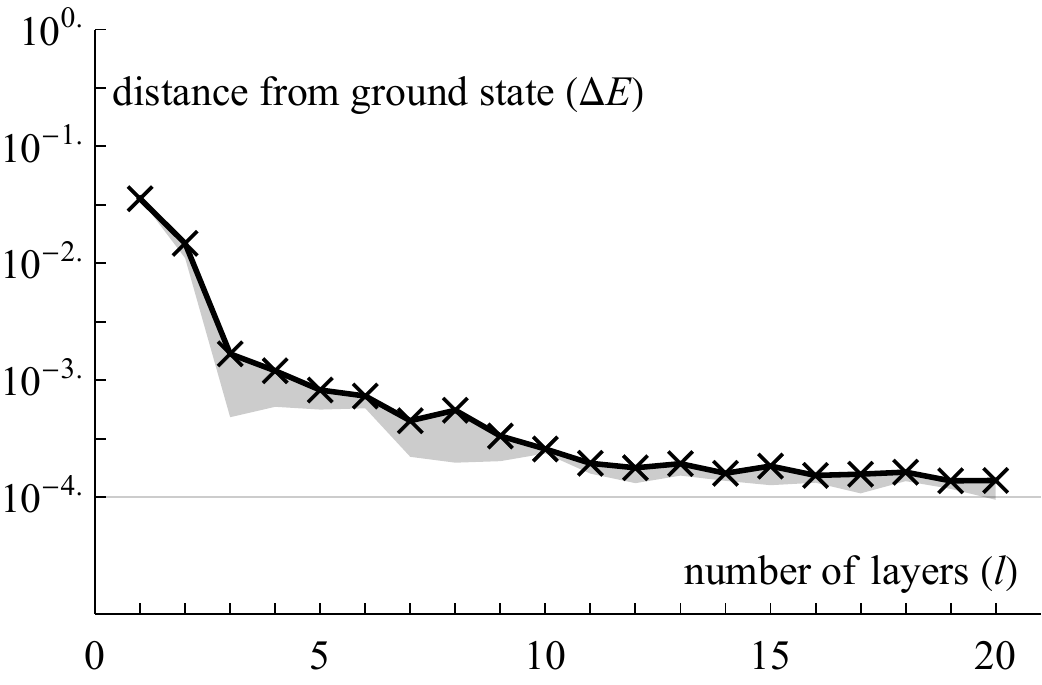}
		\caption{
				(left) Reducing errors in the derangement operator $D_2$ via generalised twirling operations as
				discussed in Appendix~\ref{symmetries}. Errors in measuring expectation values of 50 randomly
				selected observables $\sigma \in \{ \mathrm{Id}_2, X, Y, Z\}^{\otimes N} $ were computed
				with ($\tilde{\mathcal{E}}$) and without ($\mathcal{E}$) twirling. The median of the ratios $F$
				of the two errors $F:=\tilde{\mathcal{E}}/\mathcal{E}$ is plotted as a function of the
				number of expected errors in the derangement circuit ($\xi$) -- shading represents the quantiles $1/4$ and $3/4$.
				Twirling reduces 30-50\% of errors in the noisy derangement circuit, but it can be viewed as a worst-case
				scenario as discussed in Appendix~\ref{symmetries}.
				(right) Finding the ground state energy $E$ of the spin-ring Hamiltonian from Eq.~\ref{spin_ring}
				using the variational Hamiltonian ansatz with $l$ layers as discussed in Appendix~\ref{appendix_ground_state}.
				The ansatz parameters were optimised in 5 independent instances and the average of the distance
				from the ground state energy $\Delta E$ is plotted as a function of $l$. Shading represents the minimum
				of the 5 instances.
			\label{twirl_ground}
		}
	\end{centering}
\end{figure*}

We have simulated the above example assuming a system
of $N=3,4$ and $5$ qubits using recompiled controlled-SWAP operators from Appendix~\ref{appendix_ground_state}. We assume
the noise model and the same native gateset as in Appendix~\ref{appendix_ground_state}. We randomly generate input states $\rho$
and compute the error in the derangement
circuit as $\mathcal{E}= \tr[\rho^2 \sigma] / \tr[\rho^2] - \frac{2\mathrm{prob}_0-1}{2\mathrm{prob'}_0-1}$,
where we estimate the probabilities $\mathrm{prob}_0$ and $\mathrm{prob'}_0$ using the circuit in Fig.~\ref{schem}.
We then randomly select and apply 50 pairs of Pauli strings $P_1$ and $P_2$ and average the estimated probabilities;
we denote the resulting errors as $\tilde{\mathcal{E}}$. The ratio of the errors with and without twirling
$F:=\tilde{\mathcal{E}}/\mathcal{E}$
is plotted in Fig.~\ref{twirl_ground}(left).
Note that the controlled Pauli strings
$P_1$ and $P_2$ introduce additional noise when compared to just applying the plain controlled-derangement operator.
Nevertheless, Fig.~\ref{twirl_ground}(left) highlights that the above twirling scheme is still able to
reduce 30--50\% of errors in the practically most important region, i.e., when the circuit error rate is $\xi < 2$.

It is important to note that the simulations in Fig.~\ref{twirl_ground}(left) should be viewed as a worst-case scenario
for the following reasons.
(a) In the example we have considered $n=2$ copies, in which case we need $6$ entangling gates per qubit in a register
to fully recompile controlled-SWAP operators.
In the simulations we naively implemented the twirling technique resulting in $2$ additional entangling gates per qubit.
When we consider a larger number of copies, e.g., $n=5$, the overhead of the twirling technique remains $2$ 
entangling gates, but implementing the derangement operation requires proportionally more entangling gates. This leads to
a decreasing overhead of the twirling technique.
(b) As discussed above, the controlled-observable can be combined with elementary controlled-SWAP gates via recompilation resulting
in no increase in the number of two qubit gates in the derangement circuit. Similarly, we can recompile the entire, twirled
controlled-SWAP operator into one compact circuit. This will significantly reduce the gate-count overhead of the twirled circuit
thereby increasing the efficacy of the error reduction factor $F$.

\section{Numerical Simulations \label{numerics}}

\subsection{Simulations in Fig.~\ref{sampling}}

We consider an alternating-layer ansatz with 10 layers and 12 qubits as illustrated in Fig.~\ref{figansatz}.
The circuit consists of overall $372$ noisy gates and 
each two-qubit gate undergoes 2-qubit depolarising noise with $0.5\%$ probability while each single-qubit gate undergoes
depolarising noise with $0.05\%$ probability.
Each gate is parametrised and we have selected their parameters randomly.

We computed the density matrix of a single copy of the state $\rho$ and our derangement
circuit uses $n$ copies of this state as input. While we have numerically verified full derangement circuits
with smaller density matrices, here we aim to efficiently compute approximation errors. In particular,
we only need to store a single copy of $\rho$ and perform computations to obtain
$\tr[\rho^n \sigma]$ and $\tr[\rho^n]$ for randomly selected Pauli strings $\sigma$.
Furthermore, we diagonalise $\rho$ and use its eigenvalues for computing
R{\'e}nyi entropies exactly while we use its dominant eigenvector $| \psi \rangle$ to
determine the  expectation value $\langle \psi | \sigma | \psi \rangle$.

We remark here that we compute approximation errors as the deviation from the expectation value obtained from
the dominant eigenvector $\langle \psi | \sigma | \psi \rangle$ and do not directly compare to noise-free
computations due to the coherent mismatch discussed in Appendix~\ref{coherrorsec} (which becomes negligible for large systems
and can be addressed with standard techniques).
In our simulations this coherent mismatch was below $10^{-4}$,
and could be corrected with usual techniques that aim to suppress coherent errors
as discussed in the main text.

\subsection{Simulations in Fig.~\ref{extrapolation}}

We have simulated a derangement circuit that takes $n=3$ copies
of a noisy $4$-qubit state as input and the controlled SWAP operators 
also undergo depolarising noise (with a probability of $10^{-3}$) as shown in Fig.~\ref{figansatz2}.
The input state is produced by a parametrised $4$-qubit circuit
and we have selected $50$ sets of parameters randomly and
performed extrapolation techniques on each instance as shown in Fig.~\ref{extrapolation}.

\begin{figure*}[tb]
	\begin{centering}
		\includegraphics[width=0.6\textwidth]{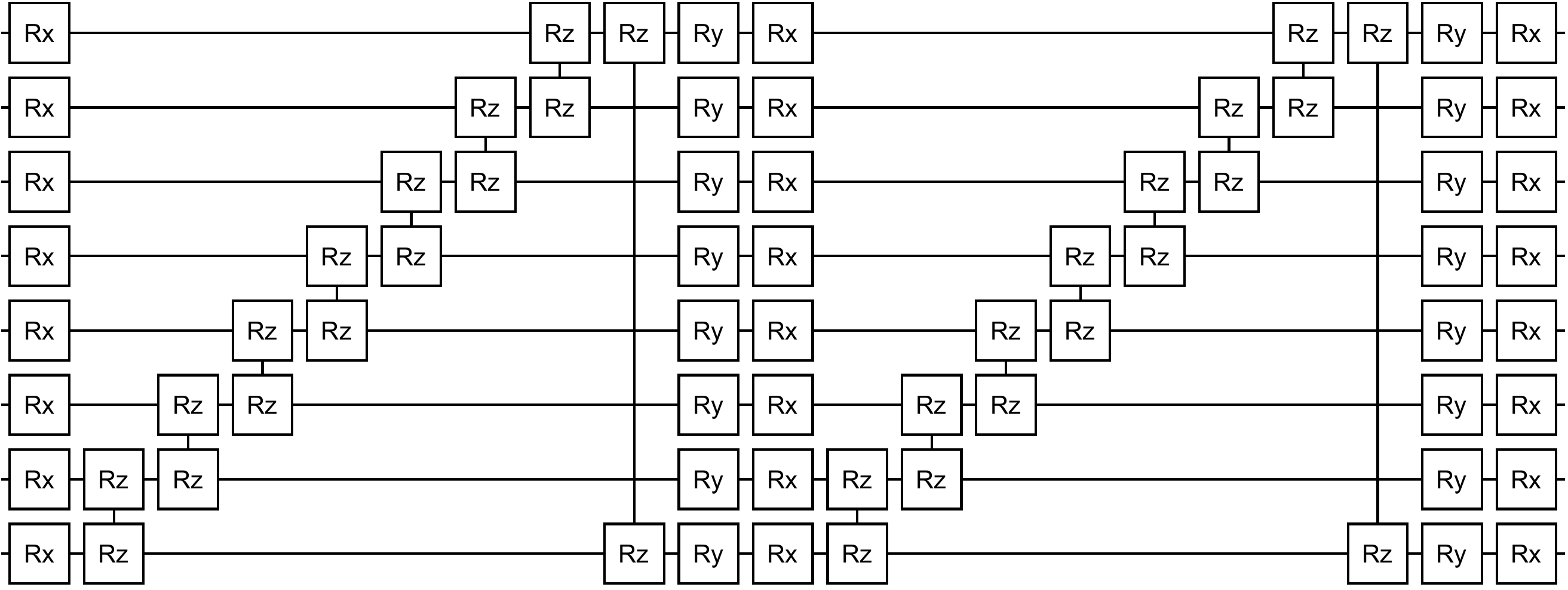}
		\caption{
			Example of a 2-block ansatz circuit of 8 qubits.
			We used a 10-block circuit of 12 qubits in our simulations
			in Fig.~\ref{sampling} consisting of overall $372$ noisy gates.
			Each two-qubit gate undergoes
			2-qubit depolarising noise with $0.5\%$ probability and each single-qubit gate undergoes
			depolarising noise with $0.05\%$ probability.
			\label{figansatz}
		}
	\end{centering}
\end{figure*}

\subsection{Ground state simulation in Fig.~\ref{groundstate}\label{appendix_ground_state}}

We consider the spin-ring Hamiltonian in Eq.~\ref{spin_ring} and aim to determine its ground state
using the Variational Hamiltonian Ansatz \cite{cerezo2020variationalreview,endo2020hybrid,bharti2021noisy,farhi2014quantum,pagano2020quantum,arute2020quantum,babbush2018low,wecker2015progress,cade2020strategies,PRXQuantum.1.020319}. The ansatz consist of alternating layers of time
evolutions under the Hamiltonians which we define as
\begin{equation*}
	\mathcal{H}_{0} :=  \sum_{k = 1 }^N \omega_k Z_k
	\quad \quad \text{and}  \quad \quad 
	\mathcal{H}_{1} := J \, \sum_{k \in \text{ring}(N)}  \vec{\sigma}_k \cdot \vec{\sigma}_{k+1},
\end{equation*}
via $\mathcal{H} = \mathcal{H}_0 + \mathcal{H}_1$,
as illustrated in Fig.~\ref{groundstate}(a).
We can analytically determine and start the optimisation from the ground state $|\psi_{init}\rangle$ of $\mathcal{H}_{0}$
as a computational basis state as in, e.g.,  \cite{pagano2020quantum}. We then apply alternating layers
of the parametrised evolutions $A(\gamma_k):= e^{-i \gamma_k \mathcal{H}_1 }$ and $B(\beta_k):=  e^{-i \beta_k \mathcal{H}_0 }$
to this initial state as
\begin{equation*}
	|\psi(\underline{\beta}, \underline{\gamma}) \rangle =  
	 B(\beta_l)  A(\gamma_l)  \dots A(\gamma_2) B(\beta_1)  A(\gamma_1)  \,	|\psi_{init}\rangle,
\end{equation*}
using overall $l$ layers. The parameters $\underline{\beta}$ and $\underline{\gamma}$ are optimised by a
classical co-processor such that the estimated energy $E:=\tr[\rho \mathcal{H}]$ is minimised. 
We consider a quantum device that can natively implement single-qubit $R_y$ and $R_z$ rotation
gates as well as XX gates of the form $\exp[-i\theta X_j X_k]$ between any pairs $j \neq k$ of qubits.
This gateset is comparable to ion-trap systems \cite{pogorelov2021compact} and the above discussed
ansatz can be implemented efficiently the following way. The evolution under $\mathcal{H}_1$ is Trotterised
such that every term in the Hamiltonian is implemented independently via a gate of the form, e.g., 
$\exp[-i\theta X_j X_k]$. Single-qubit rotations are used to implement gates of the form 
$\exp[-i\theta Z_j Z_k]$ by rotating the $Z$ basis to an $X$ basis.
It follows that the ansatz circuit with $l$ layers can be implemented via $3Nl$
applications of the $XX$ entangling gates.

	The number of layers $l$ to reach a given precision with respect to the ground state
	depends on the particular Hamiltonian and on the number of qubits \cite{PRXQuantum.1.020319,zhou2020quantum}.
	We fix every degree of freedom and determine the number of layers required to reach a difference
	$\Delta E = 10^{-4}$ to the ground-state energy. We set $N=6$ qubits, choose a coupling constant $J=0.1$  and
	randomly generate and fix the on-site energies as $\underline{\omega} =
	(-0.70983, -0.0517, 0.9065, -0.9265, 0.0950, -0.49597)$.
	We optimise the ansatz parameters
	$\underline{\beta}$ and $\underline{\gamma}$ by applying 1000 iterations of natural gradient evolution
	to a set of randomly chosen initial parameters in the vicinity of the parameters that approximate the
	adiabatic evolution.	
	We perform 5 independent optimisations and plot the average and the minimum of the distance $\Delta E$
	in Fig.~\ref{twirl_ground} (right). The average and the minimum follow the expected scaling \cite{PRXQuantum.1.020319,zhou2020quantum}
	and the minimum reaches $\Delta E = 10^{-4}$ at $l=20$ layers. Furthermore, regarding the ansatz depth we find that $l=20$ is comparable to
	results of refs.~\cite{niu2019optimizing,zhou2020quantum,PRXQuantum.1.020319}. It follows that we can implement
	the ansatz circuit with $3Nl = 360$ applications of the native entangling gate.
	Let us remark that even though one may be able to find more compact ans{\"a}tze \cite{grimsley2019adaptive}, the VHA
	has the strong benefit of being informed by
	the problem structure -- and it is guaranteed to find the ground state for an increasing depth due to its
	convergence to an adiabatic evolution \cite{farhi2014quantum}. Furthermore, even when using more compact ans{\"a}tze,
	it is generally expected that the depth of the computation needs to grow when increasing the scale of the computation
	as discussed in the main text. 

We assume the following noise model. Single-qubit gates are followed by dephasing noise with probability
$\epsilon$ and damping (relaxation) noise with a small probability $0.1\epsilon$. We also assume that
the experimentalist can amplify this noise by increasing the value of $\epsilon$. Furthermore, the qubits also
undergo a small depolarising noise with probability $0.07 \epsilon$ but we assume this noise cannot
be amplified by the experimentalist. The ratio of the non-extrapolatable noise to the total gate error rate
can be expressed via the probabilities
\begin{equation*}
	\frac{
		\mathrm{Prob}(\text{non-extrapolatable error})
	}{
		\mathrm{Prob}(\mathrm{error})
	}
= \frac{
	 0.07\epsilon
}{
	1-(1-\epsilon)(1-0.1\epsilon)(1-0.07\epsilon) 
}
\approx 0.06 + 0.009 \epsilon + \mathcal{O}(\epsilon^2).
\end{equation*}
We assume the same noise model in case of two-qubit gates but
with all probabilities magnified by a factor of $5$, i.e., $\epsilon \rightarrow 5 \epsilon$.

In selecting an error model, it was important to include key characteristics of real systems while retaining the ability to perform efficient simulations.  The model chosen exhibits the key elements of dephasing, damping and depolarising, and therefore captures the core characteristics of systems such as ion-traps, encompassing finite T2 relaxation (via dephasing), T1 relaxation (via damping) and imperfect control, heating etc. (via depolarisation). It is typical in ion traps, superconducting systems, and other platforms that single-qubit gate infidelities are significantly less severe than those of two qubit gates and therefore this characteristic was incorporated. Moreover, no real system can be expected to support perfect extrapolation as this implies flawless scaling of all error contributions; here the chosen model makes the assumption that the non-extrapolatable component is small; this is favourable to established extrapolation techniques and therefore provides a rigorous test for our new protocol.

While the resulting comparison is therefore physically plausible, it is worth noting that the ESD technique should also be expected to be robust over a wide variety of other noise models. Specifically, the theoretical error bounds in Result 1 depend only on the eigenvalue distribution of the density matrix; indeed, additional simulations were performed (not reported here) using a depolarising noise model that confirm these theoretical expectations.

\begin{figure*}[tb]
	\begin{centering}
		\includegraphics[width=0.95\textwidth]{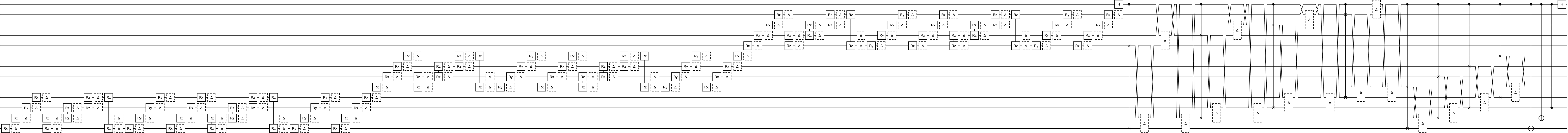}
		\caption{
			Circuit that we simulate in Fig.~\ref{extrapolation}.
			A controlled SWAP gate acting on qubits $k,l$ and $m$
			are followed by two qubit depolarisations between
			qubits $k,l$ and  $k,m$ and $l,m$. A variant in which
			damping errors follow the depolarisations effectively 
			resulted in the same error mitigation performance.			
			\label{figansatz2}
		}
	\end{centering}
\end{figure*}

In the present example we consider $N=6$ qubits and $n=2$ copies, i.e., we need to simulate the density matrix
of $13$ qubits which is equivalent to a $26$-qubit pure-state simulation.
We recompile the derangement circuit into hardware native gates as discussed in Appendix~\ref{appendix_recompile}:
we only need to recompile the elementary controlled-SWAP gates up to a local $SU(4)$ freedom (that acts on the swapped qubits).
When estimating the probability $\prob'$ we thus need to use type B circuits from the second column of Table~\ref{recompile},
see also second column in Fig~\ref{circuits1}. Implementing the corresponding derangement circuit thus requires overall $5N$ applications of the
entangling gate. When estimating the probability $\prob$ we either use circuits of type B or type C depending on whether the
observable acts on the particular qubit. For example, when estimating the expectation value of the observable $Z_1$,
we use type B circuits for all controlled-SWAP gates except the one that swaps the first qubits in both registers as
$\text{SWAP}_{1,1'}$. For the latter we use the type C circuit after rotating the basis of qubits $1$ and $1'$ so that
that the observable is effectively mapped $Z_1 \rightarrow X_1$, refer to third column in Table~\ref{recompile}
and to Fig~\ref{circuits3}. Thus estimating the probability $\prob$ for a single-qubit observable requires $5(N-1)+4 = 29$
applications of the native entangling gates while in case of two-qubit observables we need $5(N-2)+8 = 28$ entangling gates.
Note that estimating expectation values of non-local observables as $\sigma \in \{X,Y,Z\}^{\otimes N}$ would be more
noise-robust as we would only need $4N$ entangling gates to implement the derangement circuit.

Orange diamonds in Fig.~\ref{groundstate}(c) show the performance of zero noise extrapolation using a polynomial fitting:
this technique reduced close to
$85\%$ of errors in the region where the circuit error rate is not too large, i.e., $\xi < 1$. Linear and exponential fits were
also implemented, but polynomial fitting slightly outperformed exponential and linear fits.
Magenta dots in Fig.~\ref{groundstate}(c) show the performance of the noisy derangement circuit. Errors in the derangement circuit
were amplified and extrapolated using polynomial fitting, see black crosses in Fig.~\ref{groundstate}(c).
The extrapolation significantly reduces errors in the derangement circuit, thus closely approximating the performance of
noiseless derangement circuits (dashed blue line).
Note that even without extrapolation, the derangement circuit can reduce the errors by orders of magnitude, i.e., 
compare red squares with magenta dots. In all cases polynomial fitting was preformed via a least squares fitting
of a degree $3$ polynomial using $6$ estimated points in the interval between $\epsilon$ and $2\epsilon$.

	Let us finally illustrate how connectivity constraints may affect the ESD approach via the following simple
	example.
	Let us assume that qubits form a $2 \times (N+1)$ array and nearest neighbour interactions are possible.
	The qubits in positions $(0,1\rightarrow N)$ and $(1,1\rightarrow N)$ are assigned to two copies of the $N$-qubit state,
	while we assign the qubit $(0,0)$ to the ancilla. Once the main computation is
	done, one can implement the controlled-SWAP operator between qubits $(0,1)$ and $(1,1)$ controlled on the
	ancilla $(0,0)$ using only nearest neighbour interactions via Table~\ref{recompile}.
	We then move the ancilla to the next position by swapping qubits $(0,0)$ and $(0,1)$. This now allows
	us to apply the next controlled-SWAP operator between qubits $(0,2)$ and $(1,2)$ controlled on the
	ancilla $(0,1)$. Repeating this procedure for all qubits in the registers allows us to implement the
	derangement circuit with an overhead of only 1 extra two-qubit SWAP gate per controlled-SWAP operation.
	It is straightforward to generalise this idea to arbitrary numbers of copies or to ``parallelising''
	the process by distributing the ancilla among many qubits via a GHZ state. On the other hand, when
	implementing the ansatz circuit from Fig.~\ref{groundstate}, one needs to apply $2N$ SWAP gates to
	be able to entangle the first and last qubits in the register as required for the spin-ring Hamiltonian.
	This increases the number of entangling gates in the ansatz circuit from $3Nl$ to $9Nl$. Thus in such a
	scenario connectivity constraints would work in our favour.

\end{document}